\documentclass[journal,draftclsnofoot,12pt,onecolumn]{IEEEtran}
\usepackage{cite}  
\usepackage{color}
\usepackage{array}
\usepackage{algorithm} 
\usepackage{algorithmic} 
\usepackage{multirow} 
\usepackage{amsmath} 
\usepackage{xcolor}

\usepackage{graphicx}
\usepackage{epstopdf}

\usepackage{amsfonts}
\usepackage{stfloats}
\usepackage{amsthm,amsmath,amssymb}
\usepackage{mathrsfs}
\usepackage{hyperref}
\usepackage{bm}
\usepackage{amsmath}

\newtheorem{lemma}{Lemma}

\ifCLASSINFOpdf

\else

\fi

\begin{document}

\title{\huge{Joint Precoding for Active Intelligent Transmitting Surface Empowered Outdoor-to-Indoor Communication in mmWave Cellular Networks}}
\author{Xie Xie, 
        Chen He, 
       Feifei Gao,~\IEEEmembership{Fellow,~IEEE,}
        Zhu Han,~\IEEEmembership{Fellow,~IEEE,}
        \\and Z. Jane Wang,~\IEEEmembership{Fellow,~IEEE}
\thanks{Xie Xie and Chen He are with the School of Information Science and Technology, Northwest University, Xi'an, 710069, China. Corresponding author: Chen He (email: chenhe@nwu.edu.cn).
Feifei Gao is with the Department of Automation, Tsinghua University, Beijing, 100084, China.
Zhu Han is with the Department of Electrical and Computer Engineering, University of Houston, TX, USA.
Z. Jane Wang is with the Department of Electrical and Computer Engineering, The University of British Columbia, Vancouver, BC V6T1Z4, Canada.}
}
\maketitle


\begin{abstract}
Outdoor-to-indoor communications in millimeter-wave (mmWave) cellular networks have been one challenging research problem due to the severe attenuation and the high penetration loss caused by the propagation characteristics of mmWave signals. 
We propose a viable solution to implement the outdoor-to-indoor mmWave communication system with the aid of an active intelligent transmitting surface (active-ITS), where the active-ITS allows the incoming signal from an outdoor base station (BS) to pass through the surface and be received by the indoor user-equipments (UEs) after shifting its phase and magnifying its amplitude.
Then, the problem of joint precoding of the BS and active-ITS is investigated to maximize the weighted sum-rate (WSR) of the communication system. An efficient block coordinate descent (BCD) based algorithm is developed to solve it with the suboptimal solutions in nearly closed-forms. In addition, to reduce the size and hardware cost of an active-ITS, we provide a block-amplifying architecture to partially remove the circuit components for power-amplifying, where multiple transmissive-type elements (TEs) in each block share a same power amplifier. Simulations indicate that
active-ITS has the potential of achieving a given performance with much fewer TEs compared to the passive-ITS under the same total system power consumption, which makes it suitable for application to the size-limited and aesthetic-needed scenario, and the inevitable performance degradation caused by the block-amplifying architecture is acceptable.
\end{abstract}


\IEEEpeerreviewmaketitle

\section{Introduction}
Millimeter-wave (mmWave) cellular networks are capable of supplying the ever-increasing demand of rates for advanced fifth-generation communications thanks to their abundant available bandwidth.
However, a fundamental challenge for mmWave communications is that the mmWave signal experience a severe attenuation and a high penetration loss compared
with the lower frequency bands, which makes mmWave signals highly vulnerable to obstacles \cite{7400949, 7105406,6834753}.
The emerging technique of
reconfigurable intelligent surfaces (RISs) have been proposed as a promising candidate for alleviating the unfavorable properties of mmWave signals.
RIS is an ultra-thin metasurface comprising multiple programmable elements, which enables to achieve a high beamforming gain by smartly manipulating the incident signal for proactively customizing the radio propagation environment \cite{9110915,9086766,9140329}.  
More importantly, RIS can significantly reduce the outage caused by the presence of random blockages through establishing virtual line-of-sight (LoS) links between base stations (BSs) and user equipments (UEs), which can considerably enhance the reliability of mmWave communications, especially in harsh urban propagation environments \cite{9314027,9234098,9505311}.
These benefits have inspired a lot of work to investigate RIS-assisted mmWave communication networks and verify that RIS in favor of enhancing the signal strength, extending the service range, and improving the spectral- and energy-efficiency \cite{9390351,9629293,9680686,9310290}. 

On the other hand, outdoor-to-indoor communication in mmWave cellular networks is a common communication scenario as the most mobile data traffic is consumed indoors, however, which is challenging since BSs and UEs are located on the opposite side of building structures \cite{7381698}, e.g., walls and windows, while penetrating them leads a severe attenuation (measurements have shown that around 40 and 28 dB for tinted-glass and brick \cite{6655403}, respectively), which limits the feasibility for an outdoor BS communicates with indoor UEs inside the buildings.
Relay enabled system is a potential solution for outdoor-to-indoor mmWave communications but has some drawbacks, e.g., expensive hardware components and high signal processing complexity \cite{8844988}, while these shortcomings of relays can be overcome by utilizing the superiority of the RISs \cite{9514544,8930608}.
However, the widely studied RISs-assisted systems focused on traditional reflective-type RISs, also referred to intelligent reflecting surfaces (IRSs) \cite{9591503}, as objects hanging on walls or facades of buildings, which face a placement restriction: both BSs and UEs have to be located on the front side of an IRS, and hence IRSs can only achieve a half-space coverage and are not up to the outdoor-to-indoor communication. 
As a remedy, some works \cite{9513283,9362274} innovative deployed multiple IRSs in a cooperative multi-hop manner to bypass the obstacle. Nevertheless, this does not apply to
outdoor-to-indoor communications due to no matter where IRSs are placed, the signal after reflecting cannot bypass the wall to UEs behind the IRSs \cite{9514544,9591503}.




To challenge this restriction and facilitate more flexible deployment of RIS, the novel concept of transmissive(refractive)-type RIS, also named intelligent transmitting surfaces (ITS) has been proposed \cite{9365009,9598898}.
ITS is a meta-material reflection-less surface structure \cite{8058505}, which is equipped with a large number of transmissive-type elements (TEs), provides beam shaping, steering, and focus capabilities \cite{pfeiffer2013metamaterial}. A key feature of ITS is that the incoming signal can pass through each TE, and then the refracted signal can be received by UEs  behind the ITS.
More importantly, ITSs can be flexibly deployed by coexisting with existing infrastructures in the middle of a communication environment, such as embedded in walls or windows between two different environments (outdoor-to-indoor, room-to-room, etc.) \cite{9591503,9514544}. 
ITSs can fully unleash the potential of RISs on breaking the half-space (reflection space) limitation of signal propagation manipulated by IRS, and so that covers the back side space (transmission space)  
\footnote{It is worth noted that by integrating IRS and ITS together, a highly flexible full-space manipulation of signal propagation can be achieved \cite{9365009,9690478,aldababsa2021simultaneous}. In other works, the authors proposed a pair of novel resemble prototypes, simultaneously transmitting and reflecting RISs (STAR-RISs) \cite{9570143} and intelligent Omni-surface (IOSs) \cite{9491943}, where the surfaces can simultaneously transmit and reflect the incident signals. ITS or IRS can be a special full transmission or reflection mode of these concepts \cite{9690478}.}.


However, identically with the reflective cascaded channel in the IRS-assisted communication, the transmissive cascaded channel in the ITS-assisted communication will still suffer a multiplicative fading effect \cite{zhang2021active}, which potentially causes the ITS achieve a poor performance gain.
A common approach to improve the performance gain of RISs is increasing the number of reflective-type elements (REs) or TEs \cite{9427474}. Intuitively speaking, this approach is feasible for the IRS but might not be practical for the ITS due to that the number of TEs is limited by the physical size of windows and the aesthetic appearance effect, which might be the price to be paid for that the ITS is embedded with existing building structures \cite{9591503}.  
Additionally, the substantial electromagnetic penetration loss when the mmWave signal impinges upon and penetrates the ITS cannot be negligible \cite{9514544}. 
These concerns imply that the performance gain reaped by the ITS might be worthy weak. 


To this end, the deployment of an active-ITS will bring a notable leap forward for improving performance in the outdoor-to-indoor mmWave communication. 
The noticeable difference between the active-ITS and the passive-ITS is that the former can refract the incoming signals after amplifying amplitudes and shifting phases simultaneously, rather than only refracting them after changing phases as done by the latter, which offer an extra flexibility to reconfigure the incident signal. Identical to the principle of active REs \cite{9652031,9723093,9377648}, an active TE equips with an additional power amplifier for magnifying the signal amplitude \cite{9530403}. In this way, active-ITS brings an extra degree of freedom (DOF) in the beamforming, and the number of TEs of active-ITS can be beneficially reduced compared with the passive-ITS with a same performance \cite{9734027,9530750}.
The active-ITS seems to act as an active concave lens \cite{9310290} that directly refracts the incident signal with power amplification at the electromagnetic level. Nevertheless, it still takes the advantages of the passive-ITS that does not have expensive and power-consuming ratio-frequency (RF) chains, and hence has no capability for signal processing \cite{tasci2021amplifying,9733238,DBLP:journals/corr/abs-2006-06644}, which is fundamentally different from conventional amplify-and-forward (AF) relays \cite{zhang2021active}.  


Motivated by the above background, we aim to utilize the superiority of an active-ITS to alleviate the negative effects of the mmWave signal so that to implement the challenging but important outdoor-to-indoor communication system in the mmWave cellular network with a high performance. 
Particularly, the main contributions of this paper are summarized as follows:
\begin{itemize}
    \item To our best knowledge, this is the first work to implement outdoor-to-indoor mmWave communication with the aid of an active-ITS, where the incoming signal can be refracted to the back side of the active-ITS after magnifying its amplitude and manipulating its phase.
    \item To jointly optimize the precoding matrix of both the BS and active-ITS in the outdoor-to-indoor mmWave communication system, we formulate a weighted sum-rate (WSR) maximization problem, which is non-convex and challenging to obtain its corresponding globally optimal solution. 
    Consequently, we propose a compromising approach to solve it, where the original problem is transformed into an equivalent form, and then a block coordinate
    descent (BCD) based algorithm is employed to obtain the suboptimal solutions in nearly closed-forms of the linear precoding matrix of the BS, and the power amplification factor matrix and the transmissive phase-shifting matrix of the active-ITS.
    \item In order to reduce the size and hardware cost of the active-ITS, we provide a block-amplifying architecture to partially remove the circuit components for power-amplifying, where multiple TEs in each block share a same power amplifier. 
    Then, we extend the proposed algorithm under the element-amplifying architecture to jointly optimize the precoding matrix under the block-amplifying architecture. 
    \item Simulation results demonstrate that the active-ITS can significantly improve the WSR performance of the outdoor-to-indoor mmWave communication system, and the inevitable performance loss caused by the block-amplifying ITS architecture is acceptable. 
\end{itemize}

The organization of the rest of this work is organized as follows. In Section \ref{sec2}, we describe the active-ITS empowered outdoor-to-indoor communications in the mmWave cellular network model, and formulate the WSR maximization problem. In Section \ref{sec3}, we propose an efficient joint  design algorithm to solve it. Then, in Section \ref{sec3block}, the proposed algorithm is extended to solve the problem with the block-amplifying ITS architecture.
Simulation results are presented in Section \ref{sec4}, and finally the paper is concluded in Section \ref{sec5}.

\textit{Notations:} Scalars, vectors, and matrices are presented by lower-case, bold-face lower-case, and bold-face upper-case letters, respectively. $\mathcal C\mathcal N\left(\mathbf 0, \mathbf I\right)$ denotes the circularly symmetric complex Gaussian (CSCG) distribution with zero mean and covariance matrix $\mathbf I_{N}$, where $\mathbf I_{N}$ denotes an $N\times N$ identity matrix. $\operatorname{Vecd}\left\{ \cdot \right\}$ forms a vector out
of the diagonal of its matrix argument. $ \circ $ denote the Hadamard products. $\mathbf A^{\operatorname{H}}$ and $\operatorname{Tr}\left(\mathbf A\right)$ denote the Hermitian and trace operators of matrix $\mathbf A$, respectively. $\operatorname{Re}\left\{a\right\}$ is the real part of a scalar $a$.

\section{SYSTEM MODEL AND PROBLEM FORMULATION}\label{sec2}
\begin{figure}
    \centering
    \includegraphics[width=.82\linewidth]{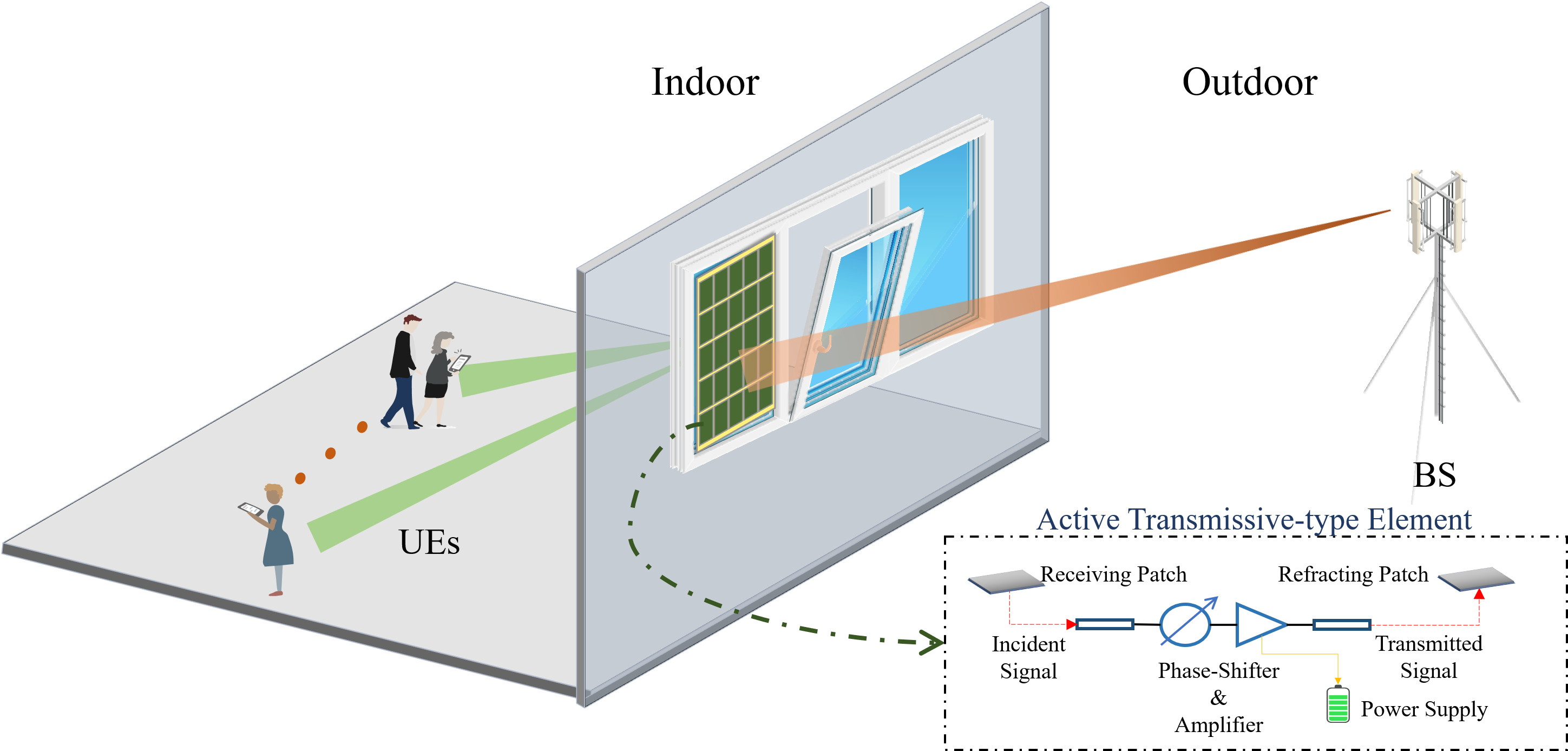}
    \caption{A system model schematic of the active-ITS empowered  outdoor-to-indoor mmWave communication.}
 \label{f1} 
\end{figure}

As shown in Fig. \ref{f1}, we consider the downlink of an outdoor-to-indoor mmWave communication scenario, where the direct-link channels from the BS to UEs are assumed severely blocked by the concrete-wall and the tinted-glass, and hence, fall into almost complete outage \cite{9390351}. Thanks to the transmissive characteristic of the ITS, the incoming mmWave signal can pass through the surface when impinge upon it, and hence, the signal can be refracted from the outdoor BS to the indoor UEs through the ITS. To alleviate the severe attenuation of mmWave signals and greatly reduce the number of TEs of the ITS for application to the size-limited window, we assume each active TE of the ITS can not only shift the phase of the incident signal, but also amplify the amplitude of the incident signal with the aid of the power supply \cite{9530403}. 
In addition, a smart controller is attached to the active-ITS and responsible to operate the active-ITS with the transmissive coefficients, which are coordinated by the BS \cite{9365009,9090356,9279253,9680675}.

\subsection{Channel Model}

Let $M_t$, $M_r$, and $N$ denote the number of the transmitting antennas, the receiving antennas, and the TEs equipped by the BS, UEs, and active-ITS, respectively. The complicated uniform planar arrays (UPA) antenna configuration is employed at the BS, UEs, and active-ITS, which is more practical than uniform linear array (ULA) for RIS-assisted systems\cite{9629293}. In addition, due to the low diffraction from objects of mmWave signals, mmWave channels are usually characterized according to the widely used Saleh-Valenzuela model \cite{9390351,9629293,9680686,9310290}. {Hence, the channel matrices from the BS to the active-ITS and from the active-ITS to the $k$-th UE are denoted by $\mathbf G \in \mathbb C^{N\times M_t}$ and $\mathbf H_k\in \mathbb C^{N\times M_r},\forall k \in K$, respectively, which can be mathematically determined as follows}
\begin{subequations}
\begin{align}
    \mathbf G&=\sqrt{\frac{M_tN}{L_{\text{BI}}}}\sum\nolimits_{l=1}^{L_{\text{BI}}}{\alpha_l\mathbf a_{\text{ITS}}\left(\nu_l^{\text{AOA}},\vartheta_l^{\text{AOA}}  \right)\mathbf a_{\text{BS}}^{\operatorname{H}}\left(\nu_l^{\text{AOD}},\vartheta_l^{\text{AOD}} \right)},\\
    \mathbf H_k&=\sqrt{\frac{M_rN}{P_{\text{IU}}}}\sum\nolimits_{p=1}^{P_{\text{IU}}}{\beta_p\mathbf a_{\text{UE}}\left(\nu_p^{\text{AOA}},\vartheta_p^{\text{AOA}}  \right)\mathbf a_{\text{ITS}}^{\operatorname{H}}\left(\nu_p^{\text{AOD}},\vartheta_p^{\text{AOD}} \right)}, \forall k\in K,
\end{align}
\end{subequations}
where $L_{\text{BI}}$ and $P_{\text{IU}}$ are the number of propagation paths between the BS and the active-ITS, and between the active-ITS and the $k$-th UE, and where variables $\alpha_l$ and $\beta_p$ denote the complex gains for links, with $l=1$ and $p=1$ stand for the LoS path and the remaining are NLoS paths. $\mathbf a_{\text{BS}}\left(\nu_l^{\text{AOD}},\vartheta_l^{\text{AOD}} \right)$ and $\mathbf a_{\text{ITS}}\left(\nu_p^{\text{AOD}},\vartheta_p^{\text{AOD}}\right)$ are transmit array response vectors of the BS and the active-ITS, $\mathbf a_{\text{UE}}\left(\nu_l^{\text{AOA}},\vartheta_l^{\text{AOA}} \right)$ and $\mathbf a_{\text{ITS}}\left(\nu_p^{\text{AOA}},\vartheta_p^{\text{AOA}}\right)$ are receive array response vectors of the UEs and the active-ITS, respectively, where $\nu^{\text{AOA}}$ ($\nu^{\text{AOD}}$) and $\vartheta^{\text{AOA}}$ ($\vartheta^{\text{AOD}}$) denote the azimuth and elevation angles of arrival (departure), respectively. The typical array response vector for the UPA model can be mathematically expressed by
\begin{align}\label{aresponse}
    \mathbf a\left(\nu,\vartheta \right)=&\frac{1}{\sqrt{HV}}\left[1,\cdots,e^{j\frac{2\pi d}{\lambda}\left(h\sin\left(\vartheta\right)\sin\left(\nu\right)+v\cos\left(\vartheta\right)\right)},\right.\notag\\
    &\qquad\qquad\qquad\qquad\cdots,\left.e^{j\frac{2\pi d}{\lambda}\left(\left(H-1\right)\sin\left(\vartheta\right)\sin\left(\nu\right)+\left(V-1\right)\cos\left(\vartheta\right)\right)}\right]^{\operatorname{T}},
\end{align}
where $H$ and $V$ denote the number of TEs (antennas) in the horizontal and vertical directions of the active-ITS (BS/UEs), and the element spacing is set to be $d=\lambda/2$.
To unveil the theoretical upper bound of performance gain achieved by the active-ITS, we follow the assumption \cite{9491943,8930608,9090356,9279253,9680675} that the full instantaneous channel state information (CSI) knowledge can be perfectly estimated at the BS and all the calculations are performed at the BS. 


\subsection{System Model}

Let $\mathbf A=\operatorname{diag}\left(a_1,a_2,\cdots,a_{N}\right)\in \mathbb R^{N\times N}$ and $\boldsymbol{\Theta}=\operatorname{diag}\left(e^{j\varphi_1},e^{j\varphi_2},\cdots,e^{j\varphi_{N}}\right)\in \mathbb C^{N\times N}$ represent the power amplification factor matrix and the transmissive phase-shifting matrix at the active-ITS with $a_n\ge1$ and $\varphi_n\in\left[0,2\pi\right)$, $\forall n\in N$. 

The signal from the BS is given by $\mathbf x=\sum_{k=1}^{K}{\mathbf W_k\mathbf s_k},\forall k \in K$, where $\mathbf s_k \in \mathbb C^{s\times 1}$ represents $s$ desired streams for the $k$-th UE and satisfies $\mathbf s_k \sim \mathcal C\mathcal N\left(\boldsymbol{0}_{ s},\mathbf I_s\right)$, and $\mathbf W_k \in \mathbb C^{M_t\times s}$ denotes the linear precoding matrix.
As a result, the signal received at the active-ITS is given by
\begin{align}
    \mathbf y_{\operatorname{ITS}}^{r}=\mathbf G\mathbf x + \boldsymbol{\upsilon}, 
\end{align}
where $\boldsymbol{\upsilon}\in \mathbb C^{N\times 1}$ is additive white Gaussian noise (AWGN) and satisfying $\boldsymbol{\upsilon}\sim \mathcal C\mathcal N\left(\boldsymbol{0}_{N},\delta^2\mathbf I_{N}\right)$. 

Then, by employing an active-ITS, the refracted signal can be denoted by
\begin{align}
    \mathbf y_{\operatorname{ITS}}^t=\kappa\mathbf A\boldsymbol \Theta\left(\mathbf G\mathbf x+\boldsymbol \upsilon\right),
\end{align}
where $\kappa$ characterizes the extent of loss corresponding to the absorption and the reflection of the signal power when the mmWave signal impinges upon and penetrates the ITS. The noise $\boldsymbol \upsilon$ after power amplifying cannot be negligible, which is significantly different from the conventional passive-ITS \cite{9723093}. 
For simplicity of presentation, let $\boldsymbol{\Phi}\triangleq  \mathbf A\circ\boldsymbol \Theta$ denotes the transmissive coefficient matrix, where each diagonal element is denoted by $\Phi_{n,n}=a_ne^{j\varphi_n},\forall n \in N$.

The signal received by the $k$-th UE can be mathematically expressed as follows 
\begin{align}
    \mathbf y_k=&\kappa\mathbf H_k^{\operatorname{H}}\mathbf\Phi\mathbf G\mathbf W_k\mathbf s_k  + \sum\nolimits_{i=1,i\ne k}^{K}{\kappa\mathbf H_k^{\operatorname{H}}\mathbf\Phi\mathbf G\mathbf W_i\mathbf s_i}+\kappa\mathbf H_k^{\operatorname{H}}\mathbf\Phi\boldsymbol \upsilon+\mathbf n_k,\forall k \in K,
\end{align}
where $\mathbf n_k\in \mathbb C^{M_r\times 1}$ represents the thermal noise at the $k$-th UE and satisfying $\mathbf n_k\sim \mathcal C\left(\boldsymbol{0}_{M_r},\sigma_k^2\mathbf I_{M_r}\right)$. 

Therefore, the signal-to-interference-plus-noise-ratio (SINR) at the $k$-th UE is given by
\begin{align}
    \boldsymbol \Gamma_k=\kappa^2\mathbf H_k^{\operatorname{H}}\mathbf\Phi\mathbf G\mathbf W_k\mathbf W_k^{\operatorname{H}}\mathbf G^{\operatorname{H}}\mathbf\Phi^{\operatorname{H}}\mathbf H_k\mathbf V_k^{\operatorname{-1}},\forall k \in K,
\end{align}
where $\mathbf V_k$ denotes the interference-plus-noise covariance matrix and can be expressed by
\begin{align}
    \mathbf V_k&=\sum\nolimits_{i=1,i\ne k}^{K}{\kappa^2\mathbf H_k^{\operatorname{H}}\mathbf\Phi\mathbf G\mathbf W_i\mathbf W_i^{\operatorname{H}}\mathbf G^{\operatorname{H}}\mathbf\Phi^{\operatorname{H}}\mathbf H_k}+\kappa^2\delta^2\mathbf H_k^{\operatorname{H}}\mathbf\Phi\mathbf\Phi^{\operatorname{H}}\mathbf H_k+\sigma_k^2\mathbf I_{M_r},\forall k \in K.
\end{align}
Accordingly, the achievable data rate of the $k$-th UE is given by 
\begin{align}
   \mathcal R _k=\log\left|\mathbf I+\boldsymbol{\Gamma_k}\right|,\forall k \in K.
\end{align}

\subsection{Problem Formulation}

In this paper, we aim to maximize the WSR of the outdoor-to-indoor mmWave communication system by joint optimizing the linear precoding matrix at the BS, i.e., $\mathbf W$, and the transmissive coefficient matrix at the active-ITS, i.e., $\boldsymbol \Phi$, while satisfying the constraints of the maximum power budgets at the BS and active-ITS, and the transmissive phase-shifting at each TE.  
Based on the above discussions, the WSR maximization problem can be formulated as follows
\begin{subequations}
\begin{align}\label{p1}
   \mathcal P_{1}: \mathop {\max}_{\mathbf W,\mathbf\Phi}  \quad &\mathcal R\left(\mathbf W,\mathbf\Phi\right)={\sum\nolimits_{k=1}^{K}{\alpha_k \log\left|\mathbf I+\boldsymbol{\Gamma_k}\right|}},\\
    \operatorname{s.t.}\quad&\sum\nolimits_{k=1}^{K}{\Vert\mathbf W_k\Vert_{\operatorname{F}}^2}\le P_{\operatorname{BS}}^{\max},\label{p11}\\
    &\sum\nolimits_{k=1}^{K}{\kappa^2\Vert\mathbf \Phi\mathbf G\mathbf W_k\Vert_{\operatorname{F}}^2}+\delta^2\kappa^2\Vert\mathbf \Phi\Vert_{\operatorname{F}}^2\le P_{\operatorname{ITS}}^{\max},\label{p12}\\
    &\angle \boldsymbol\Phi_{n,n} \in\left[0,2\pi\right),\forall n\in N,\label{p13}
\end{align}
\end{subequations}
where $\boldsymbol \alpha=\left[\alpha_1,\alpha_2,\cdots,\alpha_{K}\right]$ denotes the weighting factor vector with each element $\alpha_ k$ is a predetermined value depending on the fairness and the required quality of service for applications, and where $P_{\operatorname{BS}}^{\max}$ and $P_{\operatorname{ITS}}^{\max}$ are the budgets of the transmission power and the amplification power at the BS and the active-ITS, respectively. 

\textbf{\textit{Remark 1:}} If we remove the amplifying power budget constraint in \eqref{p12} and set $\mathbf A=\mathbf I$, Problem $\mathcal P_{1}$ reduce to the WSR maximization problem for conventional passive-RIS assisted communication systems \cite{9090356,9279253,9680675}. However, the additional constraint and the optimization of the power amplification factor matrix \cite{9374975} make Problem $\mathcal P_{1}$ even more challenging to solve.

\section{The Proposed Suboptimal Joint Precoding Algorithm}\label{sec3}
The formulated problem $\mathcal P_{1}$ is non-convex and arduous to tackle optimally, in this section, we propose an efficient joint precoding algorithm to handle it with suboptimal solutions. Particularly, to facilitate the solution development, we first reformulate the original problem as an equivalent but tractable form, and then a BCD-based method is provided to solve the transformed problem. 

\subsection{Problem Reformulation and Decomposition}

 First, to deal with the complexity of the objective function in the formulated WSR maximization problem $\mathcal P_1$, we transform it into an equivalent weighted minimum mean-square error (WMMSE) minimization problem \cite{5756489,9090356,9279253}. 
 
Let us assume that the MMSE estimator
$\mathbf U=\left\{\mathbf U_k\in \mathbb C^{M_r\times s},\forall k \in K\right\}$ are applied to UEs, and accordingly, the estimated signal vector of the $k$-th UE is given by
\begin{align}
    \mathbf {\hat s}_k=\mathbf U_k^{\operatorname{H}}\mathbf y_k, \forall k \in K.
\end{align}
Therefore, the MMSE matrix of the $k$-th UE is given by
\begin{align}
    \mathbf E_k=&\mathbb E_{\mathbf x,\boldsymbol \upsilon,\mathbf n}\left[\left(\mathbf {\hat s}_k-\mathbf s_k\right)\left(\mathbf {\hat s}_k-\mathbf s_k\right)^{\operatorname{H}}\right]\notag\\
    =&\sum\nolimits_{i=1}^{K}{\kappa^2\mathbf U_k^{\operatorname{H}}\mathbf H_k^{\operatorname{H}}\mathbf \Phi\mathbf G\mathbf W_i\mathbf W_i^{\operatorname{H}}\mathbf G^{\operatorname{H}}\mathbf \Phi^{\operatorname{H}}\mathbf H_k\mathbf U_k}+\kappa^2\delta^2\mathbf U_k^{\operatorname{H}}\mathbf H_k^{\operatorname{H}}\mathbf\Phi\mathbf\Phi^{\operatorname{H}}\mathbf H_k\mathbf U_k\notag\\
    &+\sigma_k^2\mathbf U_k^{\operatorname{H}}\mathbf U_k+\mathbf I_s-\kappa\mathbf U_k^{\operatorname{H}}\mathbf H_k^{\operatorname{H}}\mathbf \Phi\mathbf G\mathbf W_k-\kappa\mathbf W_k^{\operatorname{H}}\mathbf G^{\operatorname{H}}\mathbf \Phi^{\operatorname{H}}\mathbf H_k\mathbf U_k, \forall k \in K.
\end{align}
By introducing an auxiliary variable $\mathbf F_k \in \mathbb C^{s\times s}$ for the $k$-th UE and defining $\mathbf F=\left\{\mathbf F_k,\forall k \in K\right\}$, problem $\mathcal P_1$ can be equivalently transformed into
\begin{subequations}\label{p2s}
\begin{align}
   \mathcal P_2: \mathop {\max}_{\mathbf F,\mathbf U,\mathbf W,\mathbf\Phi}  \; &{\sum\nolimits_{k=1}^{K}{\alpha_k h_k\left(\mathbf F,\mathbf U,\mathbf W,\mathbf\Phi\right)}}\label{p2o}\\
    \operatorname{s.t.}\;&\eqref{p11}-\eqref{p13},
\end{align}
\end{subequations}
where $h_k\left(\mathbf F,\mathbf U,\mathbf W,\mathbf\Phi\right)=\log\left|\mathbf F_k\right|-\operatorname{Tr}\left(\mathbf F_k\mathbf E_k\right)+s,\forall k \in K$.

Although the above problem has been significantly simplified compared with problem $\mathcal P_{1}$, it is still challenging to tackle. Fortunately, based on the fact that the objective function in \eqref{p2o} is concave with respect to any one of the four variables (i.e., $\mathbf U$, $\mathbf F$, $\mathbf W$, and $\boldsymbol \Phi$) when the other three variables being fixed, which makes the problem more amenable.
In the following, we provide an efficient BCD-based algorithm for solving problem $\mathcal P_2$ in an iterative manner. 

Note that the MMSE decoding matrix $\mathbf U$ and the auxiliary matrix $\mathbf F$ do not exist in any constraints in problem $\mathcal P_2$, which implies that with fixed $\mathbf W$ and $\boldsymbol\Theta$, the calculation of $\mathbf F^{\star}$ and $\mathbf U^{\star}$ constitute a pair of unconstrained optimization
problems. Therefore, the optimal solutions can be obtained by setting the first-order partial derivatives of the objective function of problem $\mathcal P_2$ with respect to $\mathbf U_k,\forall k \in K$ and $\mathbf F_k,\forall k \in K$ to be zeros, respectively. After some matrix manipulations, the optimal closed-form solutions can be respectively determined by
\begin{align}\label{uk}
    \mathbf U_k^{\star}=\left(\mathbf {\bar V}_k\right)^{\operatorname{-1}}\mathbf H_k^{\operatorname{H}}\mathbf \Phi\mathbf G\mathbf W_k,\forall k \in K,
\end{align}
where $\mathbf {\bar V}_k=\mathbf { V}_k+{\kappa^2\mathbf H_k^{\operatorname{H}}\mathbf\Phi\mathbf G\mathbf W_k\mathbf W_k^{\operatorname{H}}\mathbf G^{\operatorname{H}}\mathbf\Phi^{\operatorname{H}}\mathbf H_k}$, and
\begin{align}\label{fk}
    \mathbf F_k^{\star}=\left( \mathbf E_k^{\star}\right)^{\operatorname{-1}},\forall k\in K,
\end{align}
where $\mathbf E_k^{\star}$ is determined by 
\begin{align}
    \mathbf E_k^{\star}=\mathbf I_s-\mathbf W_k^{\operatorname{H}}\mathbf G^{\operatorname{H}}\boldsymbol\Phi^{\operatorname{H}}\mathbf H_k\mathbf {\bar V}_k^{\operatorname{-1}}\mathbf H_k^{\operatorname{H}}\boldsymbol\Phi\mathbf G\mathbf W_k, \forall k \in K.
\end{align}

Once both $\mathbf U^{\star}$ and $\mathbf F^{\star}$ are determined, the remaining work is to optimize the linear precoding matrix $\mathbf W$ and the transmissive coefficient matrix $\mathbf \Phi$. Note that $\mathbf W$ and $\mathbf \Phi$ are intricately coupled in $\mathbf E_k, \forall k\in K$, we consider decomposing problem $\mathcal P_2$ into two subproblem with respect to these two variables. In particular, the two subproblems can be expressed by
\begin{subequations}
\begin{align}
    \mathcal P_{3}: \mathbf W^{\star}=\arg\mathop {\min}_{\mathbf W} &{\sum_{k=1}^{K}{\alpha_k\operatorname{Tr}\left(\mathbf F_k\mathbf E_k \right)}},\operatorname{s.t.}\;\eqref{p11}\;\text{and}\;\eqref{p12},\label{p2S}\\
    \mathcal P_{4}: \mathbf \Phi^{\star}=\arg\mathop {\min}_{\mathbf \Phi} &{\sum_{k=1}^{K}{\alpha_k\operatorname{Tr}\left(\mathbf F_k\mathbf E_k \right)}},\operatorname{s.t.}\;\eqref{p12}\;\text{and}\;\eqref{p13}.\label{p3S}
\end{align}
\end{subequations}

In the rest of this section, we propose a pair of methods to solve the above two subproblems.

\subsection{Linear Precoding Matrix  Optimization}
In this subsection, we focus on the subproblem state in \eqref{p2S} for optimizing linear precoding matrix $\mathbf W^{\star}$ at the BS with the fixed transmissive coefficient matrix $\boldsymbol \Phi$. Since the part terms of the objective function in problem \eqref{p2S}, i.e.,  $\operatorname{Tr}\left({\alpha_k\kappa^2\delta^2\mathbf U_k^{\operatorname{H}}\mathbf H_k^{\operatorname{H}}\mathbf\Phi\mathbf\Phi^{\operatorname{H}}\mathbf H_k\mathbf U_k}\right)$, and $\operatorname{Tr}\left({\sigma_k^2\mathbf U_k^{\operatorname{H}}\mathbf U_k}\right)+\mathbf I_s$ are irrelevant constants with respect to $\mathbf W$,
and hence, by omitting them, the optimal solution of $\mathbf W$ can be determined by solving the following subproblem
\begin{subequations}\label{p21}
\begin{align}
  \mathcal P_{3\mbox{-}1}:  \mathop {\min}_{\mathbf W}  \quad &    \sum_{k=1}^{K}{\operatorname{Tr}\left({\mathbf W_k^{\operatorname{H}}\mathbf Q\mathbf W_k}\right)}
    -\sum_{k=1}^{K}{\operatorname{Tr}\left({\alpha_k\kappa\mathbf F_k\mathbf U_k^{\operatorname{H}}\mathbf H_k^{\operatorname{H}}\mathbf \Phi\mathbf G\mathbf W_k}\right)}\notag\\
    &
    -\sum_{k=1}^{K}{\operatorname{Tr}\left({\alpha_k\kappa\mathbf F_k\mathbf W_k^{\operatorname{H}}\mathbf G^{\operatorname{H}}\mathbf \Phi^{\operatorname{H}}\mathbf H_k\mathbf U_k}\right)},\label{p21o}\\
    \operatorname{s.t.}\quad&\sum_{k=1}^{K}{\operatorname{Tr}\left({\mathbf W_k^{\operatorname{H}}\mathbf W_k}\right)}\le P_{\operatorname{BS}}^{\max},\label{p211}\\
    &\sum_{k=1}^{K}{\operatorname{Tr}\left({\kappa^2\mathbf W_k^{\operatorname{H}}\mathbf G^{\operatorname{H}}\mathbf \Phi^{\operatorname{H}}\mathbf \Phi\mathbf G\mathbf W_k}\right)}\le \hat P_{\operatorname{ITS}}^{\max},\label{p212}
\end{align}
\end{subequations}
where $\hat P_{\operatorname{ITS}}^{\max}\triangleq P_{\operatorname{ITS}}^{\max}-\operatorname{Tr}\left(\delta^2\kappa^2\boldsymbol{\Phi}^{\operatorname{H}}\boldsymbol{\Phi}\right)$, 
and $\mathbf Q=\sum_{k=1}^{K}{\alpha_k\kappa^2\mathbf G^{\operatorname{H}}\boldsymbol{\Phi}^{\operatorname{H}}\mathbf H_k\mathbf U_k\mathbf F_k\mathbf U_k^{\operatorname{H}}\mathbf H_k^{\operatorname{H}}\boldsymbol{\Phi}}\mathbf G$.

It can be proved that $\mathbf Q$ is a positive semi-definite matrix, therefore the objective function in \eqref{p21o} is a convex function. In addition, the constraints in \eqref{p211} and \eqref{p212} are convex with respect to $\mathbf W$, and so problem $\mathcal P_{3\mbox{-}1}$ constitutes a convex optimization problem, which can be effectively solved by employing convex solver toolbox, e.g., CVX \cite{cvx}. Instead of relying on the generic solver with high computational complexity, we provide the following duality method to further proceed with problem $\mathcal P_{3\mbox{-}1}$ efficiently.  

By introducing a pair of auxiliary variables $\lambda \ge 0$ and $\mu \ge 0$, the constraints in \eqref{p211} and \eqref{p212} can be combined as
\begin{align}
   &\lambda\sum_{k=1}^{K}{\operatorname{Tr}\left({\mathbf W_k^{\operatorname{H}}\mathbf W_k}\right)}+\mu\sum_{k=1}^{K}{\operatorname{Tr}\left({\kappa^2\mathbf W_k^{\operatorname{H}}\mathbf G^{\operatorname{H}}\mathbf \Phi^{\operatorname{H}}\mathbf \Phi\mathbf G\mathbf W_k}\right)}\le P_{\operatorname{sum}}^{\max},\label{p22wa}
\end{align}
where the constant $P_{\operatorname{sum}}^{\max}$ is defined by $P_{\operatorname{sum}}^{\max}\triangleq \lambda P_{\operatorname{BS}}^{\max}+\mu \hat P_{\operatorname{ITS}}^{\max}$. 
Consequently, problem $\mathcal P_{3\mbox{-}1}$ can be recast by 
\begin{align}\label{p22}
  \mathcal P_{3\text{-}2}:  g\left(\lambda,\mu\right)\triangleq\left\{\mathop {\min}_{\mathbf W}\,\eqref{p21o},\; \operatorname{s.t.}\, \eqref{p22wa}\right\}.
\end{align}

It is evident that a feasible solution of problem $\mathcal P_{3\mbox{-}1}$ is also feasible for problem $\mathcal P_{3\mbox{-}2}$. The relation of the optimal values between the two problems is shown in the following result.

\begin{lemma}
The optimal value of problem $\mathcal P_{3\text{-}2}$ for any given pair of $\left\{\lambda,\mu\right\}$, $\lambda\ge0$, $\mu\ge0$, is an lower bound on the optimal value of problem $\mathcal P_{3\text{-}1}$. Furthermore,  the optimal value of the dual problem, i.e., $\mathop {\max}_{\lambda, \mu}\,g\left(\lambda,\mu\right)$, is equal to that of problem $\mathcal P_{3\text{-}1}$.
\end{lemma}
\begin{proof}
This lemma follows \textit{Propositions} 4 and 5 in \cite{6169204}, and hence, is omitted here for brevity.
\end{proof}

With the fixed variables $\left\{\lambda,\mu\right\}$, by introducing a Lagrangian multiplier $\varepsilon\ge0$ associated with the combined constraint in \eqref{p22wa}, the corresponding Lagrangian function can be represented by
\begin{align}
   \mathcal L\left(\mathbf W,\varepsilon\right)=&\sum_{k=1}^{K}{\operatorname{Tr}\left({\mathbf W_k^{\operatorname{H}}\mathbf {\hat Q}\left(\varepsilon\right)\mathbf W_k}\right)}
    -\sum_{k=1}^{K}{\operatorname{Tr}\left({\alpha_k\kappa\mathbf F_k\mathbf U_k^{\operatorname{H}}\mathbf H_k^{\operatorname{H}}\mathbf \Phi\mathbf G\mathbf W_k}\right)}
    \notag\\&-\sum_{k=1}^{K}{\operatorname{Tr}\left({\alpha_k\kappa\mathbf F_k\mathbf W_k^{\operatorname{H}}\mathbf G^{\operatorname{H}}\mathbf \Phi^{\operatorname{H}}\mathbf H_k\mathbf U_k}\right)}-\varepsilon P_{\operatorname{sum}}^{\max},
\end{align}
where $\mathbf {\hat Q}\left(\varepsilon\right)\triangleq\mathbf Q+\varepsilon\left(\lambda\mathbf I_{M_t}+\mu\kappa^2\mathbf G^{\operatorname{H}}\mathbf \Phi^{\operatorname{H}}\mathbf \Phi\mathbf G\right)$.

Accordingly, the Karush–Kuhn–Tucker (KKT) conditions of problem $\mathcal P_{3\text{-}2}$ are given by
\begin{subequations}
\begin{align}
    \nabla_{\mathbf W_{k}}{\mathcal L\left(\mathbf {\tilde W},\varepsilon\right)}= 2\frac{\partial \mathcal L\left(\mathbf {\tilde W},\varepsilon\right)}{\partial {\mathbf W_{k}^{\ast}}}&=\boldsymbol{0}, \forall k,\label{lar1}\\
    \varepsilon f\left(\mathbf W_k\left(\varepsilon\right)\right)&=0, \forall k,\label{lar2}
\end{align}
\end{subequations}
where $f\left(\mathbf W_k\left(\varepsilon\right)\right)=\lambda\sum_{k=1}^{K}{\operatorname{Tr}\left({\mathbf W_k^{\operatorname{H}}\mathbf W_k}\right)}
   +\mu\sum_{k=1}^{K}{\operatorname{Tr}\left({\kappa^2\mathbf W_k^{\operatorname{H}}\mathbf G^{\operatorname{H}}\mathbf \Phi^{\operatorname{H}}\mathbf \Phi\mathbf G\mathbf W_k}\right)}- P_{\operatorname{sum}}^{\max}$.

From the first KKT condition in \eqref{lar1}, with a fixed $\varepsilon$, the optimal solution of $\mathbf W_{k},\forall k$ is calculated by   
\begin{align}
    \mathbf W_k\left(\varepsilon\right)=\alpha_k\kappa\mathbf {\hat Q}\left(\varepsilon\right)^{\operatorname{-1}} \mathbf G^{\operatorname{H}}\mathbf \Phi^{\operatorname{H}}\mathbf H_k\mathbf U_k\mathbf F_k, \forall k \in K.
\end{align}
Then, the dual variable $\varepsilon$ should be determined for satisfying the second KKT condition in \eqref{lar2}. We search $\varepsilon$ for two cases, i.e., $\varepsilon=0$ and $\varepsilon> 0$. For the former case, if the inverse of $\mathbf { Q}$ exists and $f\left(\mathbf W_k\left(0\right)\right) \le 0$ holds, the optimal solution for the $k$-th UE can be obtained by 
\begin{align}
  \mathbf W_k=\mathbf W_k\left(0\right)=\alpha_k\kappa\mathbf { Q}^{\operatorname{-1}} \mathbf G^{\operatorname{H}}\mathbf \Phi^{\operatorname{H}}\mathbf H_k\mathbf U_k\mathbf F_k, \forall k \in K.
\end{align}
Otherwise, we have to find $\varepsilon$ for ensuring $f\left(\mathbf W_k\left(\varepsilon\right)\right) = 0$ satisfied. Note that $f\left(\mathbf W_k\left(\varepsilon\right)\right)$ is a monotonically decreasing function of $\varepsilon$ that enables a bisection search method to find $\varepsilon$ \cite{9090356}.
Then, the remaining work is to solve the following dual problem
\begin{align}
    \mathop {\max}_{\lambda, \mu}\,g\left(\lambda,\mu\right).
\end{align}

We employ the iterative subgradient method \cite{9279253} to find the solutions. First, the subgradient directions of the above problem are determined by 
\begin{subequations}
\begin{align}
    f_{P_{\text{BS}}}=\sum_{k=1}^{K}{\operatorname{Tr}\left({\mathbf W_k^{\operatorname{H}}\mathbf W_k}\right)}-P_{\text{BS}}^{\max},\\
    f_{P_{\text{ITS}}}=\sum_{k=1}^{K}{\operatorname{Tr}\left({\kappa^2\mathbf W_k^{\operatorname{H}}\mathbf G^{\operatorname{H}}\mathbf \Phi^{\operatorname{H}}\mathbf \Phi\mathbf G\mathbf W_k}\right)}-\hat P_{\text{ITS}}^{\max}.
\end{align}
\end{subequations}
    
Then, the values of $\lambda$ and $\mu$ can be updated in an iterative manner, the details are summarized in Algorithm \ref{a1}, where $\xi $ denotes the step size of the subgradient algorithm, $\epsilon$ represents the tolerance, and the superscript stands for the number of iteration. The overall dual-based method for optimizing the linear precoding matrix at the BS is summarized in Algorithm \ref{a2}, and the optimal solution $\mathbf W^{\star}$ can be obtained when the algorithm converged.
\begin{algorithm}[t]
\caption{Iterative Subgradient algorithm}
\label{a1} 
\begin{algorithmic}[1]
\STATE Initialize: $\lambda^{\operatorname{0}}$, $\mu^{\operatorname{0}}$, and iteration index $p=0$. 
\REPEAT 
\STATE 
$\lambda^{{p+1}}$$\leftarrow\lambda^{{p}}+\xi f_{P_{\text{BS}}}$,\\
\STATE 
$\mu^{{p+1}}$  $\leftarrow\mu^{{p}}+\xi f_{P_{\text{ITS}}}$.
\UNTIL $\left|\lambda^{{p+1}}-\lambda^{{p}}\right|\le \epsilon$ and $\left|\mu^{{p+1}}-\mu^{{p}}\right|\le \epsilon$.
\end{algorithmic} 
\end{algorithm}

\begin{algorithm}[t]
\caption{Dual-based Algorithm for Calculating $\mathbf W$}
\label{a2} 
\begin{algorithmic}[1]
\STATE Initialize: $\lambda^{\operatorname{0}}$, $\mu^{\operatorname{0}}$, $\mathbf W^{\operatorname{0}}$, $\varepsilon^{\operatorname{0}}$, and iteration index $r=0$.
\STATE Calculate the value of the function in \eqref{p21o} as $v\left(\mathbf W^{\operatorname{0}}\right)$. 
\REPEAT 
\STATE Calculate $\mathbf W^{\operatorname{r+1}}$ and search $\varepsilon^{\operatorname{r+1}}$,\\
 Update $\lambda^{\operatorname{r+1}}$ and $\mu^{\operatorname{r+1}}$ by employing Algorithm \ref{a1},
\UNTIL $|v\left(\mathbf W^{\operatorname{r+1}}\right)-v\left(\mathbf W^{\operatorname{r}}\right)|/v\left(\mathbf W^{\operatorname{r}}\right)\le\varepsilon$.
\STATE\textbf{Output:} $\mathbf W^{\star}=\mathbf W^{\operatorname{r}}$.
\end{algorithmic} 
\end{algorithm}

\subsection{Transmissive Coefficient Matrix  Optimization}
Now, we consider to solve the subproblem of optimizing the transmissive coefficient matrix at the active-ITS. With the fixed $\mathbf W^{\star}$ by employing Algorithm \ref{a2} and define $\mathbf {\tilde W}=\sum_{i=1}^{K}{\mathbf W_i\mathbf W_i^{\operatorname{H}}}$, after dropping the constant terms, problem $\mathcal P_4$ can be simplified expressed by
\begin{align}
 \mathcal P_{4}':   \mathop {\min}_{\mathbf \Phi}  \; &\sum_{k=1}^{K}{\alpha_k\kappa^2\operatorname{Tr}\left(\mathbf F_k\mathbf U_k^{\operatorname{H}}\mathbf H_k^{\operatorname{H}}\mathbf \Phi\mathbf G\mathbf {\tilde W}\mathbf G^{\operatorname{H}}\mathbf \Phi^{\operatorname{H}}\mathbf H_k\mathbf U_k\right)}\notag\\
 &+\sum_{k=1}^{K}{\alpha_k\kappa^2\delta^2\operatorname{Tr}\left(\mathbf F_k\mathbf U_k^{\operatorname{H}}\mathbf H_k^{\operatorname{H}}\mathbf \Phi\mathbf \Phi^{\operatorname{H}}\mathbf H_k\mathbf U_k\right)}\notag\\
    &-\sum_{k=1}^{K}{\alpha_k\kappa\operatorname{Tr}\left(\mathbf F_k\mathbf U_k^{\operatorname{H}}\mathbf H_k^{\operatorname{H}}\mathbf \Phi\mathbf G\mathbf W_k\right)}\notag\\
 &-\sum_{k=1}^{K}{\alpha_k\kappa\operatorname{Tr}\left(\mathbf F_k\mathbf W_k^{\operatorname{H}}\mathbf G\mathbf \Phi^{\operatorname{H}}\mathbf H_k\mathbf U_k\right)}\\
    \operatorname{s.t.}\;&\eqref{p12} \;\operatorname{and}\; \eqref{p13}.\notag
\end{align}

Note that the above problem is still difficult to tackle. In the following, we transform problem $\mathcal P_{4}'$ into an equivalent form by employing some further algebraic manipulations. 
First, by defining $\mathbf D_k=\alpha_k\mathbf H_k\mathbf U_k\mathbf F_k\mathbf U_k^{\operatorname{H}}\mathbf H_k^{\operatorname{H}}$, $\mathbf D=\sum_{k=1}^{K}{\mathbf D_k}$, and $\mathbf T = \kappa^2\mathbf G\mathbf {\tilde W}\mathbf G^{\operatorname{H}}$, we have
\begin{subequations}
\begin{align}
    \sum_{k=1}^{K}{\alpha_k\kappa^2\operatorname{Tr}\left(\mathbf F_k\mathbf U_k^{\operatorname{H}}\mathbf H_k^{\operatorname{H}}\mathbf \Phi\mathbf G\mathbf {\tilde W}\mathbf G^{\operatorname{H}}\mathbf \Phi^{\operatorname{H}}\mathbf H_k\mathbf U_k\right)}&=\operatorname{Tr}\left(\mathbf \Phi^{\operatorname{H}}\mathbf D\mathbf \Phi\mathbf T\right),\\
    \sum_{k=1}^{K}{\alpha_k\operatorname{Tr}\left(\mathbf F_k\mathbf U_k^{\operatorname{H}}\mathbf H_k^{\operatorname{H}}\mathbf \Phi\mathbf \Phi^{\operatorname{H}}\mathbf H_k\mathbf U_k\right)}&=\operatorname{Tr}\left(\mathbf \Phi^{\operatorname{H}}\mathbf D\mathbf \Phi\right),\\
    \sum_{k=1}^{K}{\kappa^2\operatorname{Tr}\left({\mathbf W_k^{\operatorname{H}}\mathbf G^{\operatorname{H}}\mathbf \Phi^{\operatorname{H}}\mathbf \Phi\mathbf G\mathbf W_k}\right)}&=\operatorname{Tr}\left(\mathbf \Phi^{\operatorname{H}}\mathbf \Phi\mathbf T\right).
\end{align}
\end{subequations}
Then, by defining $\mathbf C_k=\alpha_k\kappa\mathbf H_k\mathbf U_k\mathbf F_k\mathbf W_k^{\operatorname{H}}\mathbf G$ and $\mathbf C=\sum_{k=1}^{K}{\mathbf C_k}$, we have 
\begin{subequations}
\begin{align}
    \sum_{k=1}^{K}{\alpha_k\kappa\operatorname{Tr}\left(\mathbf F_k\mathbf U_k^{\operatorname{H}}\mathbf H_k^{\operatorname{H}}\mathbf \Phi\mathbf G\mathbf W_k\right)}&=\operatorname{Tr}\left(\mathbf C^{\operatorname{H}}\mathbf \Phi\right),\\
    \sum_{k=1}^{K}{\alpha_k\kappa\operatorname{Tr}\left(\mathbf F_k\mathbf W_k^{\operatorname{H}}\mathbf G\mathbf \Phi^{\operatorname{H}}\mathbf H_k\mathbf U_k\right)}&=\operatorname{Tr}\left(\mathbf \Phi^{\operatorname{H}}\mathbf C\right).
\end{align}
\end{subequations}
Further, we define a sequence of equalities as follows
\begin{subequations}
\begin{align}
    \operatorname{Tr}\left(\mathbf \Phi^{\operatorname{H}}\mathbf D\mathbf \Phi\mathbf T\right)&\triangleq \boldsymbol{\phi}^{\operatorname{H}}\mathbf \Omega\boldsymbol{\phi},\\ \operatorname{Tr}\left(\mathbf \Phi^{\operatorname{H}}\mathbf C\right)&\triangleq\boldsymbol{\phi}^{\operatorname{H}}\mathbf c,\\
    \operatorname{Tr}\left(\mathbf \Phi\mathbf C^{\operatorname{H}}\right)&\triangleq\mathbf c^{\operatorname{H}}\boldsymbol{\phi},
\end{align}
\end{subequations}
where $\boldsymbol \Omega = \mathbf D\circ \mathbf T^{\operatorname{T}}$, $\boldsymbol\phi=\operatorname{Vecd}\left(\boldsymbol\Phi\right)$, and $\mathbf c=\operatorname{Vecd}\left(\mathbf  C\right)$, with $\operatorname{Vecd}\left(\mathbf X\right)$ forms a vector out of the diagonal of its matrix argument. The above equalities given above follow from the properties in \cite[Theorem 1.11]{2017Matrix}.


Here, our object is to obtain the transmissive coefficient vector of the active-ITS, $\boldsymbol \phi$, which consists of power
amplification factor and phase-shifting, i.e., $\mathbf a$ and $\boldsymbol \theta$, where 
$\mathbf a=\operatorname{Vecd}\left(\mathbf  A\right)$ and $\boldsymbol \theta=\operatorname{Vecd}\left(\boldsymbol \Theta\right)$, respectively. Consequently, we have $\boldsymbol \phi=\mathbf a\circ \boldsymbol \theta$.

According to the fact that $\operatorname{Tr}\left(\boldsymbol{\Phi}^{\operatorname{H}}\boldsymbol{\Phi}\right)=\mathbf a^{\operatorname{T}}\mathbf a$, we have 
\begin{subequations}
\begin{align}
    \operatorname{Tr}\left(\boldsymbol{\Phi}^{\operatorname{H}}\boldsymbol{\Phi}\right)&=\sum_{n=1}^{N}{a_n^2},\\
    \operatorname{Tr}\left(\mathbf \Phi^{\operatorname{H}}\mathbf D\mathbf \Phi\right)&=\sum_{n=1}^{N}{ a_n^2 D_{\operatorname{n,n}}},\\
    \operatorname{Tr}\left(\mathbf \Phi^{\operatorname{H}}\mathbf \Phi\mathbf T\right)&=\sum_{n=1}^{N}{ a_n^2T_{\operatorname{n,n}}}.
\end{align}
\end{subequations}
In addition, the constraint on the amplification power budget in \eqref{p13} can be recast by
\begin{align}\label{apbn}
    \sum_{n=1}^{N}{a_n^2 T_{\operatorname{n,n}}}+\delta^2\kappa^2\sum_{n=1}^{N}{a_n^2}\le P_{\operatorname{ITS}}^{\max}.
\end{align}

Based on the above discussion, problem $\mathcal P_4'$ can be equivalently transformed as follows
\begin{subequations}\label{p31p}
\begin{align}
   \mathcal P_{4\mbox{-}1}: \mathop{\min} _{\mathbf a, \boldsymbol \theta} \;&  f\left(\mathbf a,\boldsymbol\theta\right)=\left(\mathbf a^{\operatorname{T}}\circ\boldsymbol{\theta}^{\operatorname{H}}\right)\mathbf \Omega\left(\mathbf a  \circ  \boldsymbol{\theta}\right)+\delta^2\kappa^2\sum_{n=1}^{N}{a_n^2 D_{\operatorname{n,n}}}\notag\\
   &-\left(\mathbf a^{\operatorname{T}}\circ\boldsymbol{\theta}^{\operatorname{H}}\right)\mathbf c-\mathbf c^{\operatorname{H}}\left(\mathbf a  \circ  \boldsymbol{\theta}\right),\label{p31}\\
    \operatorname{s.t.}\;& \sum_{n=1}^{N}{ a_n^2\left( T_{\operatorname{n,n}}+\delta^2\kappa^2\right)}\le P_{\operatorname{ITS}}^{\max},\label{p311}\\
    &\angle \theta_n=\varphi_n\in\left[0,2\pi\right),\forall n \in N. \label{p312}
\end{align}
\end{subequations}

Although problem $\mathcal P_{4-1}$  has been significantly simplified compared with problem $\mathcal P_4'$, it is still non-convex. As a result, the dual-based algorithm cannot be employed here due to the dual gap is not zero. Therefore, we consider transform the above challenging problem into an equivalent form by employing a price mechanism \cite{9110849}:
\begin{subequations}
\begin{align}
   \mathcal P_{4\mbox{-}1}': \mathop{\min} _{\mathbf a, \boldsymbol \theta} \;& h\left(\mathbf a,\boldsymbol\theta\right)= f\left(\mathbf a,\boldsymbol\theta\right)+\eta g\left(\mathbf a\right)\\
    \operatorname{s.t.}\;& \eqref{p312}.
\end{align}
\end{subequations}
where $\eta \ge 0$ is the introduced price of the function $g\left(\mathbf a\right)=\sum_{n=1}^{N}{ a_n^2\left( T_{\operatorname{n,n}}+\delta^2\kappa^2\right)}$, which is the left hand of constraint in \eqref{p311}. The value of $\eta$ can be obtained by employing the bisection search method similar to search $\varepsilon$ \cite{9110849}.

Then, for a given $\eta$, we propose an efficient element-wise alternate sequential optimization (ASO) algorithm \cite{8930608, 9110912} to obtain the high-quality suboptimal solutions of the power amplification factor vector $\mathbf a$ and the transmissive phase-shifting vector $\boldsymbol\theta$ \cite{9374975}. 

More specifically, we further express the terms of the objective function $h\left(\mathbf a,\boldsymbol\theta\right)$ as follows  
\begin{align}
   &\left(\mathbf a^{\operatorname{T}}\circ\boldsymbol{\theta}^{\operatorname{H}}\right)\mathbf \Omega\left(\mathbf a  \circ  \boldsymbol{\theta}\right) \notag\\
   =& \sum_{i=1,i\ne n}^{N}{\left(\mathbf a^{\operatorname{T}}  \circ  \boldsymbol{\theta}^{\operatorname{H}} \right)\boldsymbol\Omega_{:,i} a_i\theta_i}+\left(\mathbf a^{\operatorname{T}}  \circ  \boldsymbol{\theta}^{\operatorname{H}} \right)\boldsymbol\Omega_{:,n} a_n\theta_n\notag\\
   =&\sum_{i=1,i\ne n}^{N}{ a_n\theta_n^*\Omega_{n,i} a_i\theta_i}+\sum_{j=1,j\ne n}^{N}{\sum_{i=1,i\ne n}^{N}{ a_j\theta_j^*\Omega_{j,i} a_i\theta_i}}+ a_n\theta_n^*\Omega_{n,n} a_n\theta_n+\sum_{i=1,i\ne n}^{N}{ a_i\theta_i^*\Omega_{i,n} a_n\theta_n},
\end{align}
and
\begin{align}
    \left(\mathbf a^{\operatorname{T}}\circ\boldsymbol{\theta}^{\operatorname{H}}\right)\mathbf c =a_n\theta_n^*c_n +\sum_{i=1,i\ne n}^{N}{ a_i\theta_i^* c_i}.
\end{align}

It can be readily verified that $\mathbf \Omega = \mathbf D\circ \mathbf T^{\operatorname{T}}$ is a positive semi-definite matrix due to $\mathbf D$ and $\mathbf T^{\operatorname{T}}$ are semi-definite matrices, and hence, we have $\Omega_{n,i}=\Omega_{i,n}^*$. Similarly, the constraint on the amplification power budget in \eqref{apbn} can be expressed by $a_n^2\left( T_{\operatorname{n,n}}+\delta^2\kappa^2\right)\le {\bar P}_{\operatorname{ITS}}^{\max}$
where ${\bar P}_{\operatorname{ITS}}^{\max}=P_{\operatorname{ITS}}^{\max}-\sum_{i=1,i\ne n}^{N}{a_i^2\left( T_{{i,i}}+\delta^2\kappa^2\right)}$. Clearly, the objective function $h\left(\mathbf a,\boldsymbol\theta\right)$ can be recast as an equivalent function with respect to the transmissive coefficients of $n$-th active TE, i.e., $a_n$ and $\theta_n$, which is given by 
\begin{align}
  f\left(a_n, \theta_n\right)=& 2\operatorname{Re}\left\{  a_n  \theta_n^*\left(\sum_{i=1,i\ne n}^{N}{\Omega_{n,i} a_i\theta_i}- c_n\right) \right\}+a_n^2\Omega_{n,n}+\delta^2\kappa^2{ a_n^2 D_{\operatorname{n,n}}}+\eta a_n^2\left( T_{\operatorname{n,n}}+\delta^2\kappa^2\right)+\chi_n,
\end{align}
where  
$\chi_n$ is a constant value with respect to the $n$-th pair of optimization variables $\left\{a_n, \theta_n\right\}$, which is given by
\begin{align}
    \chi_n=\sum_{j=1,j\ne n}^{N}{\sum_{i=1,i\ne n}^{N}{ a_j\theta_j^*\Omega_{j,i} a_i\theta_i}}+\delta^2\kappa^2\sum_{i=1,i\ne n}^{N}{a_i^2 D_{\operatorname{i,i}}}+2\operatorname{Re}\left\{ \sum_{i=1,i\ne n}^{N}{a_i \theta_i^*c_i} \right\}+\eta \sum_{i=1,i\ne n}^{N}{a_i^2\left( T_{{i,i}}+\delta^2\kappa^2\right)}.
\end{align}

Consequently, we can only investigate the following problem for sequentially optimizing a pair of values $\left\{a_n,\theta_n\right\}$ while fixing the remaining $N-1$ pairs.
By omitting the constant term $\chi_n$ in $f\left( a_n, \theta_n\right)$, which has no impact on optimizing $\left\{a_n,\theta_n\right\}$, and defining $e_n=\sum_{i=1,i\ne n}^{N}{\Omega_{n,i} a_i\theta_i}-c_n$, we have the following surrogate problem with respect to $\left\{a_n,\theta_n\right\}$, and takes the forms
\begin{subequations}
\begin{align}
    \mathcal P_{4\text{-}2}: \mathop{\min} _{a_n, \theta_n} \;&  2\operatorname{Re}\left\{a_n \theta_n^*e_n \right\}+ a_n^2\Omega_{n,n}+\delta^2\kappa^2{ a_n^2 D_{\operatorname{n,n}}}+\eta a_n^2\left(T_{\operatorname{n,n}}+\delta^2\kappa^2\right),\label{p32}\\
    \operatorname{s.t.}\;& \eqref{p312} .\label{p321}
\end{align}
\end{subequations}

Note that $\theta_n$ only appear in the constraints of \eqref{p312}, we decompose problem $\mathcal P_{4\text{-}2}$ into two subproblems as follows
\begin{subequations}
\begin{align}
      \mathcal P_{4\text{-}3\text{-}1}&:\mathop{\min} _{\theta_n}\;\operatorname{Re}\left\{ a_n \theta_n^*e_n \right\},\,\operatorname{s.t.}\eqref{p312},\\
      \mathcal P_{4\text{-}3\text{-}2}&:\mathop{\min} _{ a_n}\;\eqref{p32},\label{p331}
\end{align}
\end{subequations}
In particular, an equivalent expression for problem $\mathcal P_{4\text{-}3\text{-}1}$ is given by
\begin{subequations}
\begin{align}
    \mathop{\min}_{\varphi_n}\quad& \cos\left(-\varphi_n+\angle e_n\right),\label{eq39a} \\
    \operatorname{s.t.}\quad& \varphi_n\in\left[0,2\pi\right).
\end{align}
\end{subequations}
The optimal solution of the above problem is \begin{align}\label{phase}
    \varphi_n=\angle e_n-\pi.
\end{align}

Note that Problem $\mathcal P_{4\text{-}3\text{-}2}$ is a unconstrained optimization problem, and with the obtained solution of $\varphi_n$ in \eqref{phase}, the value of \eqref{eq39a} can be calculated as $-1$, and hence, the term of $2\operatorname{Re}\left\{a_n \theta_n^*e_n \right\}$ in \eqref{p32} reduce to $-2a_n$.
Consequently, the solution to $\mathcal P_{4\text{-}3\text{-}2}$ is given in a semi-closed-form expression as follows
\begin{align}\label{an}
    a_n=\frac{1}{\delta^2\kappa^2D_{n,n}+\Omega_{n,n}+\eta\left(T_{n,n}+\delta^2\kappa^2
    \right)}.
\end{align}


Based on the above discussions, a pair of $\left\{ a_n, \theta_n\right\}, n=1,\cdots,N$ can be conducted by successively optimizing with the other $N-1$ pairs being fixed, and then repeat until the convergence is attained. The details of the element-wise ASO algorithm are summarized in Algorithm \ref{a3}. It can be readily verified Algorithm \ref{a3} is guaranteed to converge and obtain a high quality suboptimal solution \cite{7946256}.


\begin{algorithm}[t]
\caption{Element-wise ASO Algorithm for Calculating $\boldsymbol\Phi$} 
\label{a3} 
\begin{algorithmic}[1] 
\STATE \textbf{Initialize} $\mathbf a^0$,  $\boldsymbol\theta^0$, and iteration index $l=0$.
\STATE Let the value of the function in \eqref{p31} as $\rho^0\left(\mathbf a^0, \boldsymbol\theta^0\right)$. 
\REPEAT 
\STATE \textbf{Sequentially Optimizing:} $ \mathbf a_n^{l+1}$ and $\boldsymbol\varphi_n^{l+1}$, $\forall n \in N$, by using \eqref{an} and \eqref{phase}, respectively;\\
\STATE \textbf{Calculate} $\rho^{l+1}\left(\mathbf a^{l+1}, \boldsymbol\theta^{l+1}\right)$;
\UNTIL
$\left|\rho^{l+1}-\rho^{l}\right|\le\varepsilon$, \textbf{output} $ \mathbf a^{\star}\triangleq  \mathbf a^{l}$, $ \boldsymbol\varphi^{\star}\triangleq \boldsymbol\varphi^{l}$.
\end{algorithmic} 
\end{algorithm}

\subsection{Convergence and Complexity Analysis}

The detailed description of the proposed BCD-based joint precoding algorithm is summarized in Algorithm \ref{a4}, we update one of the variables while the others being fixed. Algorithm \ref{a4} guarantees to converge which can be proved as follows.
For each updating step, with the fixed other variables, it can be readily proved that the value of the objective function $\mathcal R\left(\mathbf W, \boldsymbol\Phi\right)$ is monotonically increasing after each updating step. Additionally, since the power budget constraints at the BS and active-ITS, $\mathcal R$ is upper bounded by a finite value, thus Algorithm \ref{a4} guarantees to converge.

Then, let us now briefly analyze the total computational complexity of Algorithm \ref{a4}. The complexities for updating $\mathbf U$ and $\mathbf F$ in steps 4 and 5 are $\mathcal O\left(KM_r^3\right)$ and $\mathcal O\left(K s^3\right)$, respectively. For step 6 that update the linear precoding matrix of the BS by employing Algorithm \ref{a2}, the main complexity contribution of Algorithm \ref{a2} is calculating the inverse operation, which is on the order of $\mathcal O\left(K M_t^3\right)$, and the number of iterations required for employing Algorithm \ref{a2} is $\mathcal I_{\operatorname{Dual}}$. The complexity of Algorithm \ref{a3} is $\mathcal O\left(N^2\right)$ for each iteration, and the number of iterations is $\mathcal I_{\operatorname{ASO}}$. Therefore, the total complexity of Algorithm \ref{a4} is $\mathcal O\left(\mathcal I_{\operatorname{BCD}}\left(KM_r^3+K s^3+\mathcal I_{\operatorname{Dual}}K M_t^3+\mathcal I_{\operatorname{ASO}}N^2\right)\right)$, where $\mathcal I_{\operatorname{BCD}}$ is the total number of iterations of Algorithm \ref{a4}.

\begin{algorithm}[t]
\caption{BCD-based Joint Precoding Algorithm for Solving Problem $\mathcal P_1$} 
\label{a4} 
\begin{algorithmic}[1]
\STATE \textbf{Initialize}    
    $\mathbf {W}^{\left ( \operatorname{0} \right )}$,
    $\mathbf A ^{\left ( \operatorname{0} \right )}$, $\boldsymbol \Theta^{\left ( \operatorname{0} \right )}$,and iteration index $t=0$.
\REPEAT
\STATE \textbf{Calculate} $\boldsymbol\Phi^{\left ( \operatorname{t} \right )}=\mathbf A^{\left ( \operatorname{t} \right )}\circ \boldsymbol\Theta^{\left ( \operatorname{t} \right )}$;
\STATE \textbf{Updating} $\mathbf {U} ^{\left ( \operatorname{t+1} \right )}$ by using \eqref{uk};
\STATE \textbf{Updating} $\mathbf {F} ^{\left ( \operatorname{t+1} \right )}$, by using \eqref{fk};
\STATE \textbf{Updating} $\mathbf {W} ^{\left ( \operatorname{t+1} \right )}$, by employing Algorithm \ref{a2};
\STATE \textbf{Updating} $\boldsymbol \Theta^{\left ( \operatorname{t+1}\right )}$ and $\mathbf A^{\left ( \operatorname{t+1}\right )}$, by employing Algorithm \ref{a3};\\
\UNTIL
{${{\left| {{\mathcal{R}^{\left( \operatorname{t+1} \right)}}- {\mathcal{R}^{\left( \operatorname{t}\right)}}} \right|}} < \varepsilon$}, \textbf{output} $\mathbf {W}^{\star} \triangleq \mathbf {W}  ^{\left ( \operatorname{t} \right )}$, $\mathbf {A}^{\star} \triangleq \mathbf {A}  ^{\left ( \operatorname{t} \right )}$, and $\boldsymbol{ \Theta}^{\star} \triangleq \boldsymbol{\Theta}^{\left (  \operatorname{t} \right )}$.
\end{algorithmic} 
\end{algorithm}

\section{The Proposed Block-Amplifying architecture of active-ITS}\label{sec3block}

In this section, we provide a block-amplifying architecture of active-ITS to partially remove the circuit components for power-amplifying, which is beneficial for reducing the surface size and hardware cost. 
More specifically, as illustrated in Fig. \ref{f3}, the total active TEs are divided into several blocks, and the TEs in each block are assumed to share a same power amplifier. 
Compared with element-amplifying architecture, the block-amplifying architecture can reduce the scale of the power-amplifying circuit, which makes it further suitable for application to the size-limited scenario. In addition, by partially removing the circuit components for power amplifying, the total power consumption of the block-amplifying architecture reduce to $ NP_{\operatorname{SW}}+R P_{\operatorname{DC}}$, where $R$ ($R\le N$) denotes the number of blocks,
$P_{\operatorname{SW}}$ represent the power consumption for switching phase-shifting of each TE, and $P_{\operatorname{DC}}$ is the direct current biasing power for power-amplifying of each block \cite{9734027}.
Clearly, the power consumption of the block-amplifying architecture is less than that of the element-amplifying architecture, i.e., $N\left(P_{\operatorname{SW}}+P_{\operatorname{DC}}\right)$, which means that the more power will be used to perform the amplitude-magnifying, while it can compensate for some inevitable performance loss caused by such a architecture. For example, we consider a block-amplifying ITS with $N=60$ ,$R=10$, $P_{ITS}^{\max}=15$ dBm (31.6 mW), and follow the setting in \cite{9377648,9734027} that $P_{\operatorname{SW}}=-10$ dBm (0.1 mW) and $P_{\operatorname{DC}}=-5$ dBm (0.316 mW), and hence the total power consumption of the active-ITS for the element-amplifying architecture and block-amplifying architecture are 24.96 mW and 9.16 mW, and corresponding the power used for performing amplitude-magnifying are $6.64$ mW and $22.44$ mW, respectively.   


\begin{figure}
    \centering
    \includegraphics[width=.5\linewidth]{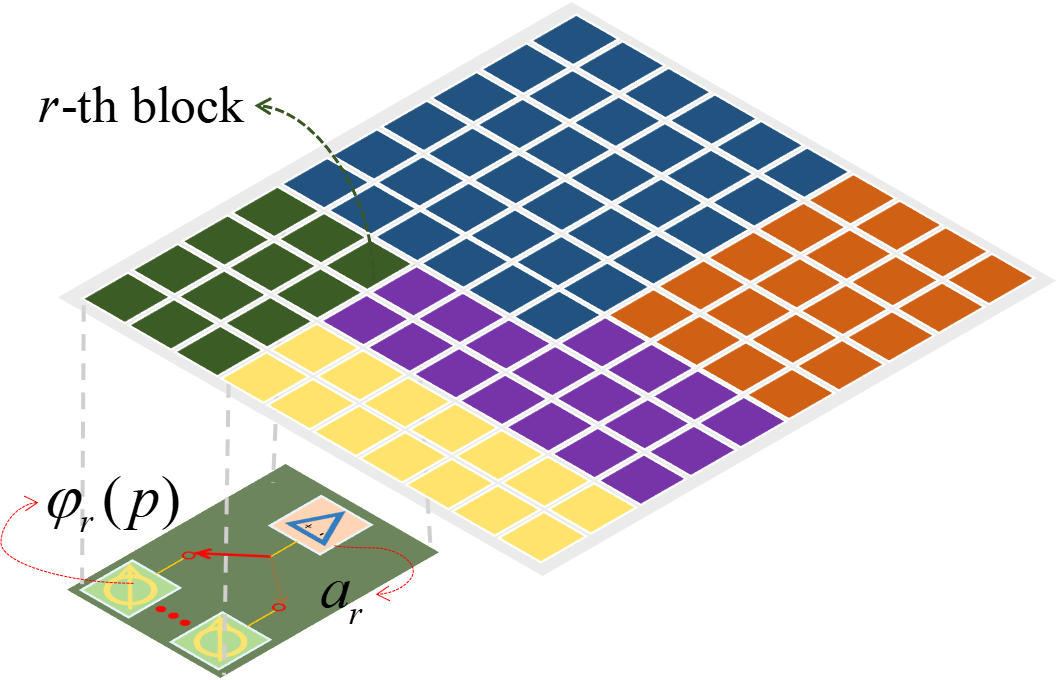}
    \caption{Block-amplifying architecture of active-ITS.}
 \label{f3} 
\end{figure}

Note that the variables $\mathbf U$, $\mathbf F$, $\mathbf W$, and $\mathbf \Theta$ can be directly obtained as that under element-amplifying architecture case, and hence, the remaining work is to optimize the power amplifying factor matrix, i.e., $\mathbf A$, under the block-amplifying architecture case. In the following, we extend the proposed ASO algorithm from the element-wise form into the block-wise form, which can optimize $\mathbf A$ in a block-by-block manner.

\textbf{\textit{Remark 2:}} From a mathematical point of view, since the element-amplifying architecture can be interpreted as a special case ($R=N$) of the block-amplifying architecture, the proposed ASO Algorithm under the block-wise form is a superset of the element-wise form.

Particularly, we partition the power amplifying factor $\mathbf {a}$ into $R$ blocks, i.e., $\mathbf {a}=\left [\mathbf { a}_1^{\operatorname{T}},\mathbf {a}_2^{\operatorname{T}}, \cdots,\mathbf { a}_R^{\operatorname{T}}\right]\in \mathbb R^{N\times 1}$, where $\mathbf { a}_r\in \mathbb R^{{\hat N}_r\times 1}$ is the power amplifying factor vector corresponding to ${\hat N}_r$ TEs in the $r$-th block, which is given by
\begin{align}\label{hata}
    \mathbf {a}_r=\left[ {{\mathbf  a}_{r}\left(1\right),{\mathbf a}_{r}\left(2\right),
    \cdots,{\mathbf a}_{r}\left({\hat N}_r\right)}\right]^{\operatorname{T}}={a}_r\boldsymbol 1_{{\hat N}_r}^{\operatorname{T}},
\end{align}
where $\mathbf a_r\left({m}\right)$ is the $\left(\sum_{l=1}^{r-1}{{\hat N}_l}+m\right)$-th element of $\mathbf a$ with $r=2,3,\cdots,R$ and $m=1,2,\cdots,{\hat N}_r$.

Here, let us revisit problem $\mathcal P_{4\text{-}1}$ stated in \eqref{p31p}. First, $\left(\mathbf a^{\operatorname{T}}\circ\boldsymbol{\theta}^{\operatorname{H}}\right)\mathbf \Omega\left(\mathbf a  \circ  \boldsymbol{\theta}\right)$ can be expanded as
\begin{align}\label{b2b}
   &\left(\mathbf a^{\operatorname{T}}\circ\boldsymbol{\theta}^{\operatorname{H}}\right)\mathbf \Omega\left(\mathbf a  \circ  \boldsymbol{\theta}\right) \notag\\
   =& \sum_{i=1,i\ne r}^{R}{\left(\mathbf a^{\operatorname{T}}  \circ  \boldsymbol{\theta}^{\operatorname{H}} \right)\boldsymbol\Omega_{i}\left(\mathbf a_i\circ\boldsymbol\theta_i\right)}+\left(\mathbf a^{\operatorname{T}}  \circ  \boldsymbol{\theta}^{\operatorname{H}} \right)\boldsymbol\Omega_{r}\left(\mathbf a_r\circ\boldsymbol\theta_r\right)\notag\\
   =&\sum_{i=1,i\ne r}^{R}{\left(\mathbf a_i^{\operatorname{T}}\circ\boldsymbol\theta_i^{\operatorname{H}}\right)\boldsymbol\Omega_{i,r}\left(\mathbf a_r\circ\boldsymbol\theta_r\right)}+\sum_{i=1,i\ne r}^{R}{\sum_{j=1}^{R}{\left(\mathbf a_j^{\operatorname{T}}\circ\boldsymbol\theta_j^{\operatorname{H}}\right)\boldsymbol\Omega_{j,i}\left(\mathbf a_i\circ\boldsymbol\theta_i\right)}}+\left(\mathbf a_r^{\operatorname{T}}\circ\boldsymbol\theta_r^{\operatorname{H}}\right)\boldsymbol\Omega_{r,r}\left(\mathbf a_r\circ\boldsymbol\theta_r\right),\notag\\
   =&\sum_{i=1,i\ne r}^{R}{\left(\left(\mathbf a_i^{\operatorname{T}}\circ\boldsymbol\theta_i^{\operatorname{H}}\right)\boldsymbol\Omega_{i,r}\left(\mathbf a_r\circ\boldsymbol\theta_r\right)+{\left(\mathbf a_i^{\operatorname{T}}\circ\boldsymbol\theta_i^{\operatorname{H}}\right)\boldsymbol\Omega_{i,r}\left(\mathbf a_r\circ\boldsymbol\theta_r\right)}\right)}
   \notag\\
   &+\sum_{i=1,i\ne r}^{R}{\sum_{j=1}^{R}{\left(\mathbf a_j^{\operatorname{T}}\circ\boldsymbol\theta_j^{\operatorname{H}}\right)\boldsymbol\Omega_{j,i}\left(\mathbf a_i\circ\boldsymbol\theta_i\right)}}+\left(\mathbf a_r^{\operatorname{T}}\circ\boldsymbol\theta_r^{\operatorname{H}}\right)\boldsymbol\Omega_{r,r}\left(\mathbf a_r\circ\boldsymbol\theta_r\right),
\end{align}
where $\boldsymbol\Omega=\left[\boldsymbol\Omega_1, \boldsymbol\Omega_2,\cdots,\boldsymbol\Omega_R \right]$ and $\boldsymbol\theta=\left[\boldsymbol\theta_1^{\operatorname{T}}, \boldsymbol\theta_2^{\operatorname{T}},\cdots,\boldsymbol\theta_R^{\operatorname{T}} \right]^{\operatorname{T}}$, with 
\begin{subequations}
\begin{align}
    \boldsymbol\Omega_r&=\left[\boldsymbol\Omega_{1,r}, \boldsymbol\Omega_{2,r},\cdots,\boldsymbol\Omega_{R,r} \right]^{\operatorname{T}}\in\mathbb C^{N\times {\hat N}_r},\\
    \boldsymbol\theta_r&=\left[\boldsymbol\theta_{r}\left(1\right), \boldsymbol\theta_{r}\left(2\right),\cdots,\boldsymbol\theta_{r\left({{\hat N}_r}\right)} \right]^{\operatorname{T}}\in\mathbb C^{ {\hat N}_r\times 1},
\end{align}
\end{subequations}
where $\boldsymbol\theta_r\left({m}\right)$ denotes the $\left(\sum_{l=1}^{r-1}{{\hat N}_l}+m\right)$-th element of $\boldsymbol \theta$ with $r=2,3,\cdots,R$ and $m=1,2,\cdots,{\hat N}_r$.

By employing the property $\boldsymbol\Omega_{i,r}=\boldsymbol\Omega_{i,r}^{\operatorname{H}}$
, \eqref{b2b} can be rewritten as
\begin{align}\label{b2b2}
   &\left(\mathbf a^{\operatorname{T}}\circ\boldsymbol{\theta}^{\operatorname{H}}\right)\mathbf \Omega\left(\mathbf a  \circ  \boldsymbol{\theta}\right) \notag\\
   =&2\operatorname{Re}\left\{\sum_{i=1,i\ne r}^{R}{\left(\mathbf a_i^{\operatorname{T}}\circ\boldsymbol\theta_i^{\operatorname{H}}\right)\boldsymbol\Omega_{i,r}\left(\mathbf a_r\circ\boldsymbol\theta_r\right)}\right\}
   +\sum_{i=1,i\ne r}^{R}{\sum_{j=1}^{R}{\left(\mathbf a_j^{\operatorname{T}}\circ\boldsymbol\theta_j^{\operatorname{H}}\right)\boldsymbol\Omega_{j,i}\left(\mathbf a_i\circ\boldsymbol\theta_i\right)}}\notag\\
   &+\left(\mathbf a_r^{\operatorname{T}}\circ\boldsymbol\theta_r^{\operatorname{H}}\right)\boldsymbol\Omega_{r,r}\left(\mathbf a_r\circ\boldsymbol\theta_r\right).
\end{align}

Then, note that the term $\sum_{i=1,i\ne r}^{R}{\left(\mathbf a_i^{\operatorname{T}}\circ\boldsymbol\theta_i^{\operatorname{H}}\right)\boldsymbol\Omega_{i,r}\left(\mathbf a_r\circ\boldsymbol\theta_r\right)}$ in \eqref{b2b2} can be expanded by
\begin{align}
\sum_{i=1,i\ne r}^{R}{\left(\mathbf a_i^{\operatorname{T}}\circ\boldsymbol\theta_i^{\operatorname{H}}\right)\boldsymbol\Omega_{i,r}\left(\mathbf a_r\circ\boldsymbol\theta_r\right)}
=&\sum_{i=1,i\ne r}^{R}{a_ra_i\left(\boldsymbol 1_{{\hat N}_r}^{\operatorname{T}}\circ\boldsymbol\theta_r^{\operatorname{H}}\right)\boldsymbol\Omega_{r,i}\left(\boldsymbol 1_{{\hat N}_i} \circ\boldsymbol\theta_i\right)}\notag\\
=&\sum_{i=1,i\ne r}^{R}{a_ra_i\boldsymbol\theta_r^{\operatorname{H}}\boldsymbol\Omega_{r,i}\boldsymbol\theta_i}.
\end{align}
Similarly, we have 
\begin{subequations}
\begin{align}
\left(\mathbf a_r^{\operatorname{T}}\circ\boldsymbol\theta_r^{\operatorname{H}}\right)\boldsymbol\Omega_{r,r}\left(\mathbf a_r\circ\boldsymbol\theta_r\right)&= a_r^2\boldsymbol\theta_r^{\operatorname{H}}\boldsymbol\Omega_{r,r}\boldsymbol\theta_r,\\
\left(\mathbf a^{\operatorname{T}}\circ\boldsymbol{\theta}^{\operatorname{H}}\right)\mathbf c 
&= a_r\boldsymbol{\theta}_r^{\operatorname{H}}\mathbf c_r  +\sum\nolimits_{i=1,i\ne r}^{R}{a_i\boldsymbol{\theta}_i^{\operatorname{H}}\mathbf c_i}.
\end{align}
\end{subequations}

By defining $\mathbf d=\operatorname{Vecd}\left(\mathbf D\right)$ and $\mathbf t=\operatorname{Vecd}\left(\mathbf T\right)$, and then partition the two vector into $R$ blocks, i.e., $\mathbf {d}=\left [\mathbf { d}_1^{\operatorname{T}},\mathbf {d}_2^{\operatorname{T}}, \cdots,\mathbf { d}_R^{\operatorname{T}}\right]\in \mathbb C^{N\times 1}$ and $\mathbf {t}=\left [\mathbf { t}_1^{\operatorname{T}},\mathbf {t}_2^{\operatorname{T}}, \cdots,\mathbf { t}_R^{\operatorname{T}}\right]\in \mathbb C^{N\times 1}$, we have
\begin{subequations}
\begin{align}
    \operatorname{Tr}\left(\boldsymbol\Phi^{\operatorname{H}}\mathbf D\boldsymbol\Phi\right)&=a_r^2 \sum_{p=1}^{{\hat N}_r}{\mathbf d_{r}\left(p\right)}+\sum_{i=1,i\ne r}^{R}{a_i^2 \sum_{q=1}^{{\hat N}_i}{\mathbf d_{i}\left(q\right)}},\\
    \operatorname{Tr}\left(\boldsymbol\Phi^{\operatorname{H}}\boldsymbol\Phi\mathbf T\right)&=a_r^2 \sum_{p=1}^{{\hat N}_r}{\mathbf t_{r}\left(p\right)}+\sum_{i=1,i\ne r}^{R}{a_i^2 \sum_{q=1}^{{\hat N}_i}{\mathbf t_{i}\left(q\right)}}.
\end{align}
\end{subequations}
Similarly, we have
\begin{align}
    \operatorname{Tr}\left(\boldsymbol\Phi^{\operatorname{H}}\boldsymbol\Phi\right)&=a_r^2{\hat N}_r+ \sum\nolimits_{i=1,i\ne r}^{R}{a_i^2 {{\hat N}_i}}.
\end{align}

Consequently, problem $\mathcal P_{4\text{-}1}$ can be reformulated by
\begin{subequations}
\begin{align}
   \mathcal P_{5}: \mathop{\min} _{a_r} \;&  f\left(a_r\right)=a_r^2m_r+2\operatorname{Re}\left\{a_r z_r\right\}+\chi_{ r},\label{psa}\\
    \operatorname{s.t.}\;& a_r^2g_r\le {\hat P}_{\operatorname{ITS}}^{\max},\label{psa1}
\end{align}
\end{subequations}
where 
\begin{subequations}
\begin{align}
    m_r&=\boldsymbol\theta_r^{\operatorname{H}}\boldsymbol\Omega_{r,r}\boldsymbol\theta_r+ \delta^2\kappa^2\sum_{p=1}^{{\hat N}_r}{\mathbf d_{r}\left(p\right)},\\ z_r&=\sum_{i=1,i\ne r}^{R}{a_i\boldsymbol\theta_r^{\operatorname{H}}\boldsymbol\Omega_{r,i}\boldsymbol\theta_i}-\boldsymbol{\theta}_r^{\operatorname{H}}\mathbf c_r,\\
    g_r&=\sum_{p=1}^{{\hat N}_r}{\mathbf t_{r}\left(p\right)}+\kappa^2\delta^2{\hat N}_r,\\
    {\hat P}_{\operatorname{ITS}}^{\max}&=P_{\operatorname{ITS}}^{\max}-\sum_{i=1,i\ne r}^{R}{a_i^2 \sum_{q=1}^{{\hat N}_i}{\mathbf t_{i}\left(q\right)}}-\kappa^2\delta^2\sum_{i=1,i\ne r}^{R}{a_i^2 {{\hat N}_i}},
\end{align}
\end{subequations}
and where $\chi_r$ denotes a constant value with respect to $a_r$ and does not affect its solution, and is given by
\begin{align}
    \chi_r=\sum_{i=1,i\ne r}^{R}{\sum_{j=1}^{R}{ a_ja_i\boldsymbol\theta_j^{\operatorname{H}}\boldsymbol\Omega_{j,i}\boldsymbol\theta_i}}+2\operatorname{Re}\left\{\sum_{i=1,i\ne r}^{R}{a_i\boldsymbol{\theta}_i^{\operatorname{H}}\mathbf c_i}\right\}+\kappa^2\delta^2\sum_{i=1,i\ne r}^{R}{a_i^2 \sum_{q=1}^{{\hat N}_i}{\mathbf d_{i}\left(q\right)}}.
\end{align}
The above problem can be solved similarly to \eqref{an} and hence it is omitted for brevity. 

Based on the above discussions, $\mathbf a_r,\forall r$ can be conducted by successively optimizing with the other $R-1$ blocks being fixed, and then repeat until the convergence is attained \cite{7946256}. 


\section{Simulation Result}\label{sec4}
\begin{figure}
    \centering
    \includegraphics[width=0.48\linewidth]{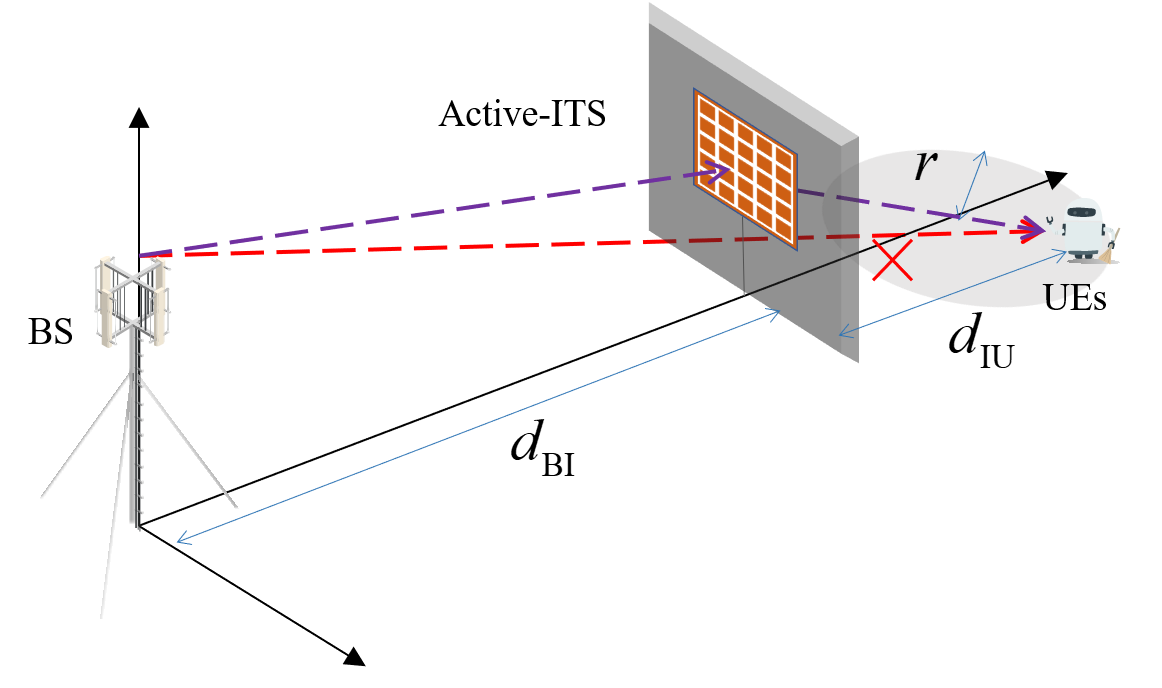}
    \caption{The simulated active-ITS empowered outdoor-to-indoor mmWave communication.}
 \label{f2} 
\end{figure}

In this section, numerical results are provided to evaluate the WSR performance achieved by the proposed BCD-based joint precoding algorithm under the two architectures of active-ITSs.

\subsection{Simulation Setup}
We consider a three-dimensional setup as illustrated in Fig. \ref{f2}, where the BS at the outdoor communicates with UEs at the indoor with the aid of an active-ITS. The BS, active-ITS, and UEs are assumed in a straight line, and heights of them are 3m, 6m, and 1.5m, respectively, and the total horizontal distances $d_{\text{BU}}=d_{\text{BI}}+d_{\text{IU}}=50$ m. The $K=4$ UEs are uniformly and randomly distributed in a
circle centered at $(0\operatorname{m},50\operatorname{m})$ with a radius of $5$ m. 



For array response vector, the azimuth angles of arrival $\upsilon^{\operatorname{AOA}}$ and departure $\upsilon^{\operatorname{AOD}}$ are uniformly distributed in the interval $\left[-\pi, \pi\right)$. The elevation angles of arrival $\vartheta^{\operatorname{AOA}}$ and departure $\vartheta^{\operatorname{AOD}}$ are uniformly distributed in the interval $\left[-\pi/2, \pi/2\right)$. 
The complex gains of LOS paths $\alpha_l,\forall l$ and $\beta_p,\forall p$ are independently distributed with $\mathcal C\mathcal N\left(0, 10^{-0.1{\operatorname{PL}}}\right)$, where the distance-dependent path loss $\operatorname{PL}$ for both the LoS and the NLoS paths can be modeled by 
\begin{align}\label{loss}
     \operatorname{PL}=\operatorname{PL_0}+10b\log_{10}\left(D\right)+\zeta \, \text{dB}, 
\end{align} 
where $D$ denotes the individual link distance and $\zeta$ denotes lognormal shadow fading following $\zeta\sim \mathcal N\left(0, \sigma_{\zeta}^2\right)$. As the real-world channel measurements for the carrier frequency of 28 GHz suggested in \cite[Table I]{6834753}, the values of the parameters in \eqref{loss} are set as $\operatorname{PL_0}=61.4$, $b=2$, and $\sigma_{\zeta}=5.8$ dB for the LoS paths, and $\operatorname{PL_0}=72$, $b=2.92$, and $\sigma_{\zeta}=8.7$ dB for NLoS paths.    
The bandwidth $B=251$ MHz, the noise power is  $\sigma^2=-174 +10\log_{10} B=-90$ dBm. 
Unless otherwise stated, the other adopted simulation parameters are set as follows: the distance between the BS and ITS of $d_{\operatorname{BI}}=45$ m, the circuit power consumption for phase-shifting and power-amplifying are $P_{\operatorname{SW}}=-10$ dBm and $P_{\operatorname{DC}}=-5$ dBm \cite{9377648}, the power budgets of the BS and active-ITS are $P_{\operatorname{BS}}^{\max}=30$ dBm and $P_{\operatorname{ITS}}^{\max}=30$ dBm, the number of TEs of $N=10\times 4$, the number of antennas of $M_t=4\times 2$ and $M_r=2\times 2$, the weighting factor
of $\alpha_k=1, \forall k$, the penetration efficiency of $\kappa=0.8$, and the target accuracy of $\varepsilon=10^{-3}$.

Without loss of generality, for the block-amplifying ITS architecture, we adopt an evenly-block strategy, where the number of TEs in each block is equal and set as 5 and 10, i.e., ``Block-amplifying ITS, $N_B=5$'' and ``Block-amplifying ITS, $N_B=10$''. We also introduce a baseline scheme ``Block-amplifying ITS, $N_B=N$'', where all TEs of ITS share the same power amplifier, so all the active TEs have identical amplitudes, i.e., $a_n=a,\forall n\in N$ \cite{9734027}. 
In addition, due to that the active-ITS requires an additional power, i.e., $P_{\operatorname{ITS}}^{\max}$, and hence, for fair comparisons, we assume that the power budget at the BS for the passive-ITS scheme (referred as ``Without-amplifying ITS'') is allocated an additional amount of $P_{\operatorname{ITS}}^{\max}$, i.e., $P_{\operatorname{BS,passive}}^{\max}=P_{\operatorname{ITS}}^{\max}+P_{\operatorname{BS,active}}^{\max}$, so that all the compared schemes have the
same total power consumption.

\begin{figure}
    \centering
    \includegraphics[width=0.5\linewidth]{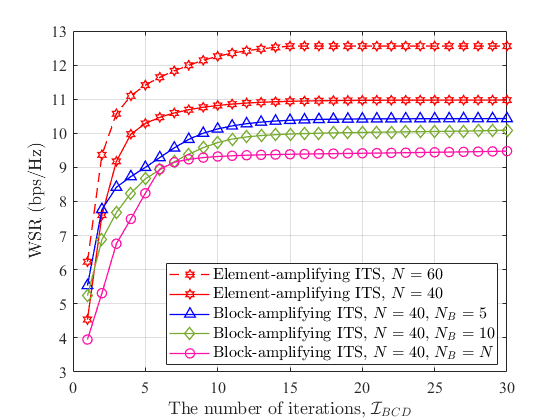}
    \caption{WSR against the number of iterations, $\mathcal I_{BCD}$.}
 \label{fiteration} 
\end{figure}

\subsection{Convergence Behavior}
In this subsection, we investigate the convergence behavior of the proposed BCD-based algorithm for the different architectures of active-ITSs with different numbers of TEs.
As shown in Fig. \ref{fiteration}, we plot the curves of the WSR performance achieved by the proposed algorithm for the active-ITS schemes (include {Element-amplifying ITS} and {Block-amplifying ITS}) against the number of iterations. 
In particular, it can be observed that the curve with $N=60$ TEs converges slower than that of $N=40$ TEs, and the curves of both the element-amplifying architecture and the block-amplifying architecture converge to corresponding stationary points after a few iterations. The observation is consistent with our expectations. It can be concluded that the proposed algorithm has good convergence behavior and the convergence speed is sensitive to the number of TEs of an active-ITS, and the impact of $N_B$ on the convergence speed is slight.

\subsection{The impact of key parameters on WSR performances}

In this subsection, we study the WSR performance achieved by the proposed algorithms for both the active-ITS scheme and the passive-ITS scheme (i.e., Without-amplifying ITS) versus the key parameters of the considered system setting. All the simulation results as follows are averaged over 1,000 independent channel realizations.

\begin{figure}
    \centering
    \includegraphics[width=0.48 \linewidth]{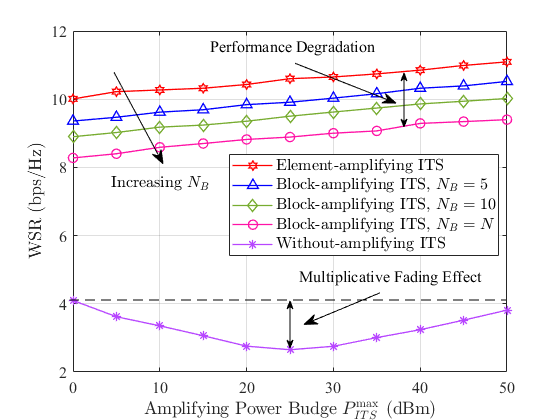}
    \caption{WSR against the position of the active-ITS, $d_{\operatorname{BI}}$.}
 \label{fdis} 
\end{figure}
In Fig. \ref{fdis}, we investigate the WSR performance for different locations of the ITS by varying the distance between the BS and ITS, i.e., $d_{\operatorname{BI}}$. One important observation is that with increasing $d_{\operatorname{BI}}$, the WSR performance achieved by the active-ITS scheme monotonically increase, which is different from the passive-ITS scheme, which achieve the highest and the worst WSR performances when passive-ITS are deployed close to the BS(UEs) and at the middle location due to the multiplicative fading effect. 
 This interesting observation indicates that the active-ITS should to be deployed closer to UEs for fully unleashing its potential on achieving a better WSR performance. 
Intuitively speaking, the deployment strategy meets the requirements of a practical outdoor-to-indoor communication scenario where the active-ITS is usually very close to the UEs, i.e., $d_{\operatorname{IU}} \ll  d_{\operatorname{BI}}$.
In addition, it can be observed that the active-ITS schemes perform constantly better than passive-ITS scheme in all locations, which implies that the active-ITS brings a significant performance gain compared with the passive-ITS due to the amplified signal is only attenuated once. 
More importantly, it can be seen that the WSR of ``Element-amplifying ITS'' scheme outperforms that of the ``Block-amplifying ITS'' schemes, which can be attributed to that the block-amplifying architecture of active-ITS inevitably lost some DoF for proactively configuring the wireless channels. This observation indicates the importance of carefully optimizing the power amplification factor matrix at the active-ITS for enhancing the WSR performance. Nevertheless, considering that such a block-amplifying architecture of active-ITS can reduce the scale of the power-amplifying circuit for application to space-limited scenarios, the resulting slight performance gap is acceptable.

The WSR performances versus the number of TEs with $P_{\operatorname{ITS}}^{\max}=30$ dBm and $P_{\operatorname{ITS}}^{\max}=20$ dBm are shown in Fig. \ref{ftes}. It illustrates that with the increasing number of TEs, both curves of the active-ITS and the passive-ITS schemes monotonically increases, since more TEs introduce more DoFs to proactively configure the wireless channel which yields a higher WSR performance. Meanwhile, it also illustrates the superiority of the active-ITS that brings a large benefit to WSR performance compared with the passive-ITS. Particularly, the usage of active ITSs can greatly reduce the number of TEs compared with passive ITS case for reaping a given performance level, and hence
greatly decrease both the size and the complexity of ITS, which is beneficial for embedding in space-limited and aesthetic-needed building structures. It can be explained by that when equipped a small number of active TEs in the active-ITS, according to the constraint in \eqref{p12}, a larger power amplification factor will be allocated at each active TE.
In addition, it can be observed that with a small number of TEs, the performance degradation between element-amplifying ITS and block-amplifying ITS is slight with $P_{\operatorname{ITS}}^{\max}=20$ dBm, while the performance gap in the large $N$ regime also is acceptable. This observation is because that for a large number of TEs at the active-ITS, the block-amplifying architecture has more power to amplify the signal, which can compensate for some of the loss of performance.
This observation highlights the effectiveness of the proposed block-amplifying architecture.

\begin{figure}
    \centering
    \includegraphics[width=0.48\linewidth]{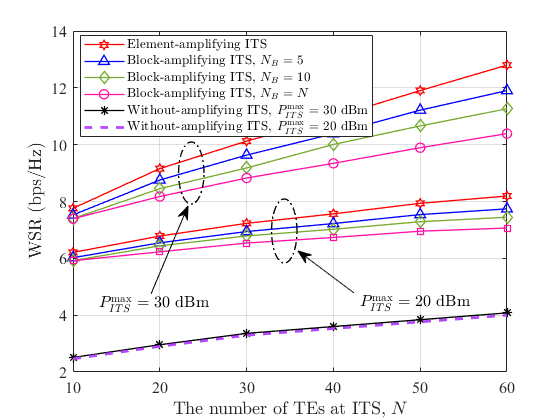}
    \caption{WSR against the number of TEs.}
 \label{ftes} 
\end{figure}

We plot the WSR performance against the amplifying power budget at ITS in Fig. \ref{fpower}. As expected, it can be observed that all the WSR performances of all schemes increase as $P_{\operatorname{ITS}}^{\max}$ increase. More specifically, the WSR performances achieved by the active-ITS schemes outperform that of passive-ITS, which indicates that the active-ITS can exploit a part of the power for fighting with the multiplicative path loss effect influence and achieving a higher WSR performance than the passive-ITS scheme for the same power consumption. In addition, it can be concluded from Fig. \ref{fpower} that the active-ITS scheme can significantly reduce the power consumption compared to the passive-ITS scheme when achieving a given performance level. More importantly, it can be observed that the performance degradation caused by the block-amplifying architecture of the active-ITS is slight when $P_{ITS}^{\max}= 15$ dBm, which is because that the more power is used for the circuit static power consumption in the small $P_{ITS}^{\max}$ regime.


\begin{figure}
    \centering
    \includegraphics[width=0.48\linewidth]{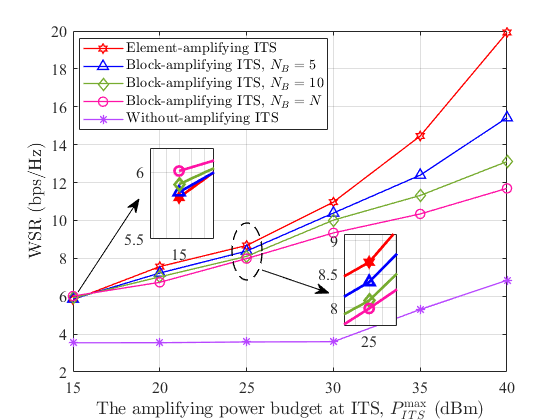}
    \caption{WSR against the amplifying power budget $P_{\operatorname{ITS}}^{\max}$ at the active-ITS.}
 \label{fpower} 
\end{figure}

\begin{figure}
    \centering 
    \includegraphics[width=0.48\linewidth]{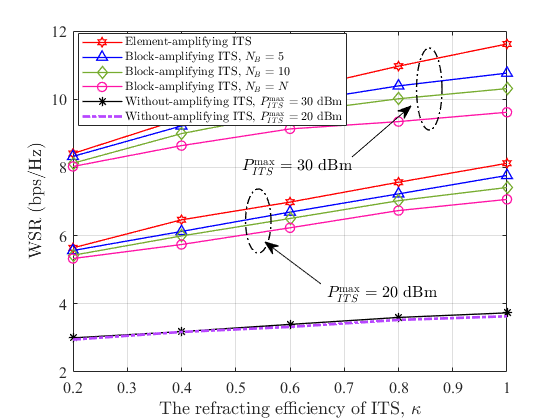}
    \caption{WSR against the refracting efficiency $\kappa$.}
 \label{fkappa} 
\end{figure}
Fig. \ref{fkappa} illustrates the WSR performance against the refracting efficiency of ITSs. 
It can be observed that the refracting efficiency of ITS has a substantial impact on the WSR performance, whereas expected, with the increasing $\kappa$, the WSR achieved by the active-ITS schemes increased significantly. 
The refracting efficiency of an ITS is determined by the hardware components, which is used to characterize the power loss caused by signal absorption and signal reflection when the mmWave signal impinges upon and penetrate the ITS. This figure demonstrates the benefits of active-ITS that can exploit the additional power amplifiers to compensate for the power loss caused by the practical hardware component.

\section{Conclusion}\label{sec5}

    In this paper, an active-ITS empowered outdoor-to-indoor mmWave communication system was investigated, where the active-ITS benefited on greatly reducing the number of TEs while maintaining a given performance compared with the passive-ITS. To jointly optimize the linear precoding matrix of the BS and transmissive coefficients of the active-ITS, we formulated a WSR maximization problem, and then a BCD-based joint precoding algorithm was proposed for solving it. In order to further reduce the size and hardware cost of active-ITS, a block-amplifying architecture was proposed to partially remove the hardware components for power amplifying. And then we extended the proposed BCD-based joint precoding algorithm into the block-wise form for optimizing the transmissive coefficients of the active-ITS under the block-amplifying architecture. Simulation results demonstrated that the active-ITS could significantly enhance the system performance, and the inevitable performance degradation caused by the block-amplifying ITS architecture was acceptable. 

\ifCLASSOPTIONcaptionsoff
  \newpage
\fi

\bibliography{ref}
\end{document}